\documentclass[10pt,journal,compsoc]{IEEEtran}
\IEEEoverridecommandlockouts
\usepackage{textcomp}
\usepackage{xcolor}
\usepackage{tikz}
\usepackage{amsmath}
\usepackage{cite}
\usepackage{amsthm}
\usepackage{amssymb}
\usepackage{graphicx}
\usepackage{extarrows}
\usepackage[noend]{algpseudocode}  
\usepackage{booktabs}
\usepackage{bookmark}
\usepackage{multirow}
\usepackage{comment}
\newtheorem{theorem}{Theorem}

\newtheorem{lemma}{Lemma}
\newtheorem{remark}{Remark}
\newtheorem{definition}{Definition}
\newtheorem{corollary}{Corollary}

\usepackage{amsfonts}
\newcommand{\e}[1]{{\mathbb E}\left[ #1 \right]}
\usepackage{color}
\usepackage{array}
\usepackage{subfigure}
\DeclareMathAlphabet{\mathcal}{OMS}{cmsy}{m}{n}
\SetMathAlphabet{\mathcal}{bold}{OMS}{cmsy}{b}{n}

\usepackage[linesnumbered,ruled,vlined]{algorithm2e}
\usepackage{enumitem}
\setenumerate[1]{itemsep=0pt,partopsep=0pt,parsep=\parskip,topsep=5pt}
\setitemize[1]{itemsep=0pt,partopsep=0pt,parsep=\parskip,topsep=5pt}
\setdescription{itemsep=0pt,partopsep=0pt,parsep=\parskip,topsep=5pt}
\usepackage{pifont}
\newcommand{\cmark}{\ding{51}}%
\newcommand{\xmark}{\ding{55}}%

\newcommand{\RH}[1]{{\color{orange}$\langle$\textsc{Runchao:} #1$\rangle$}}
\newcommand{\jn}[1]{\textcolor{black}{#1}}
\newcommand{\niu}[1]{\textcolor{black}{#1}}

\newcommand{\tabincell}[2]{\begin{tabular}{@{}#1@{}}#2\end{tabular}}  
\usepackage{multirow}
\usepackage[flushleft]{threeparttable}

\usepackage[noend]{algpseudocode}

\def\BibTeX{{\rm B\kern-.05em{\sc i\kern-.025em b}\kern-.08em
    T\kern-.1667em\lower.7ex\hbox{E}\kern-.125emX}}

\begin{document}
\interfootnotelinepenalty=10000

\author{
    Jianyu Niu$^\dagger$,
    Fangyu Gai,
    Runchao Han,
    Ren Zhang, 
    Yinqian Zhang,
    Chen Feng,
    \IEEEcompsocitemizethanks{
        \IEEEcompsocthanksitem Jianyu Niu and Yinqian Zhang are with the Engineering and Research Institute of Trustworthy Autonomous Systems and the School and Computer Science and Engineering Department, Southern University of Science and Technology, Shenzhen, China.
        Email: niujy@sustech.edu.cn and yinqianz@acm.org.

        \IEEEcompsocthanksitem Fangyu Gai and Chen Feng are with the Blockchain@UBC and the School of Engineering, The University of British Columbia (Okanagan Campus), Kelowna, BC, Canada. Email: greferry@gmail.com and chen.feng@ubc.ca. Contact author: Chen Feng.

        \IEEEcompsocthanksitem Runchao Han is with Monash University and CSIRO-Data61. E-mail: runchao.han@monash.edu.

        \IEEEcompsocthanksitem Ren Zhang is with Cryptape Co. Ltd. and Nervos. E-mail: ren@nervos.org.
    }
    \thanks{$^\dagger$ Part of the work is done at the University of British Columbia.}
}
\sloppy

\title{Crystal: Enhancing Blockchain Mining Transparency with Quorum Certificate}

\IEEEtitleabstractindextext{
    \begin{abstract}
        \jn{
            Researchers have discovered a series of theoretical attacks against Bitcoin's Nakamoto consensus; the most damaging ones are selfish mining, double-spending, and consistency delay attacks.
            These attacks have one common cause: block withholding.
            This paper proposes Crystal, which leverages quorum certificates to resist block withholding misbehavior.
            Crystal continuously elects committees from miners and requires each block to have a quorum certificate, i.e., a set of signatures issued by members of its committee.
            Consequently, an attacker has to publish its blocks to obtain quorum certificates, rendering block withholding impossible.
            To build Crystal, we design a novel two-round committee election in a Sybil-resistant, unpredictable and non-interactive way, and a reward mechanism to incentivize miners to follow the protocol.
            Our analysis and evaluations show that Crystal can significantly mitigate selfish mining and double-spending attacks.
            For example, in Bitcoin, an attacker with 30\% of the total computation power will succeed in double-spending attacks with a probability of 15.6\% to break the 6-confirmation rule; however, in Crystal, the success probability for the same attacker falls to 0.62\%.
            We provide formal end-to-end safety proofs for Crystal, ensuring no unknown attacks will be introduced.
            To the best of our knowledge, Crystal is the first protocol that prevents selfish mining and double-spending attacks while providing safety proof.
        }
    \end{abstract}
    }

\maketitle
\IEEEdisplaynontitleabstractindextext
\IEEEpeerreviewmaketitle

\IEEEraisesectionheading{\section{Introduction}\label{sec: intro}}
\IEEEPARstart{B}{itcoin}, the largest and most influential cryptocurrency as of 2022, was launched in 2009 by Satoshi Nakamoto.
The success of Bitcoin is mainly due to its innovative Nakamoto Consensus (NC) algorithm for maintaining a distributed ledger (commonly known as blockchain).
NC relies on two innovations 
: $1)$ \textit{Proof-of-Work} (PoW), in which participants, also called miners, have to solve computational puzzles to mine a block (which contains a set of transactions) without prior knowledge of identities (i.e., in a permissionless setting), and $2)$ the \textit{longest chain rule}, by which participants always mine the latest block on the longest chain and can eventually reach consensus over the same chain of blocks.
\jn{As a reward, participants can get a block reward and transaction fees for each block in the longest chain.}
NC is the first consensus algorithm that enables thousands of participants to maintain a public ledger without relying on a trusted third party.
By now, NC and its variants have been adopted in more than four hundred cryptocurrencies~\cite{mapofcoins}.

Unfortunately, the security of NC has been challenged by several attacks such as selfish mining~\cite{eyal2014majority,sapirshtein2016optimal,niu2019selfish}, double-spending~\cite{nakamoto2012bitcoin}, as well as consistency delay attack~\cite{Pass2017, Kiffer2018}.
These attacks can either undermine the incentive compatibility~\cite{eyal2014majority, sapirshtein2016optimal} or threaten the safety~\cite{nakamoto2012bitcoin, Pass2017}.
For example, in selfish mining attacks, after an attacker successfully mines two consecutive blocks on the longest chain, it can \textbf{withhold} these two blocks until it receives a block from an honest participant. In this way,
the attacker will not only get two block rewards but also make the honest block abandoned by all the participants, thereby obtaining higher revenue than its fair share.
As we can see, this attack is rooted in block withholding. More worryingly, block withholding enables a variety of other attacks~\cite{eyal2014majority,sapirshtein2016optimal,niu2019selfish, Kiffer2018}.
\jn{However, it is challenging to resist block withholding for two reasons.
First, the existing mining process is non-transparent, which allows an attacker to withhold blocks and build a private chain.
Second, opaque block propagation makes it hard to distinguish between an honest but delayed block and a maliciously withheld block.\footnote{Zhang and Preneel~\cite{Zhang2019CommonMetrics} referred to this reason as information asymmetry between the attacker and compliant participants, which is caused by unawareness of network connectivity and the lack of a globally synchronous clock.}
}

\jn{
In this paper, we propose Crystal, an NC-style protocol that \emph{resists} block withholding attacks. This distinguishes Crystal from most existing protocols that are tailored to defend only one or two specific attacks, and cannot prevent block withholding attacks completely~\cite{Pass2017fruit, perishzhang, strongchain, Bissias2017BobtailAP}. Crystal employs a novel two-round committee election in a \emph{Sybil-resistant}, \emph{unpredictable} and \emph{non-interactive} way, and enforces each block to obtain a quorum certificate (QC), i.e., a collection of the signatures from a quorum of members in the committee~\cite{castro1999practical,malkhi1997byzantine}.
This prevents the attacker from mining a private chain, as the required QC has to contain votes from honest participants who have seen the block.}
Figure~\ref{fig: architecture} presents a concrete example of the usage of QCs.
Assuming Alice, Bob, and Carol constitute the committee.
An attacker has mined a block $B$ and has obtained a signature from Bob.
To mine a block after $B$, the attacker needs a QC on $B$, which needs an extra signature from Alice or Carol.
Without the full knowledge of the committee members, the attacker has to publish $B$ to some honest participants (at least one from Alice and Carol), hoping to obtain more signatures.
However, once receiving $B$, honest participants will propagate $B$ to the network so that all participants will receive $B$.
If the attacker withholds $B$, then it will stall mining as $B$ is considered invalid without a QC.

\jn{ 
The main challenge of realizing Crystal is the committee formation in the permissionless setting.
First, the committee formation needs to resist Sybil attacks~\cite{sybil2002} such that an attacker cannot create many identities to increase its voting power in the committee.
Second, the committee formation needs to be unpredictable and non-interactive in order to prevent the attacker from predicting the committee members. 
Otherwise, the attacker can impose certain attacks, e.g., bribery, Eclipse attacks~\cite{eclipse2015}, and DoS attack~\cite{chen2021efficient}, to break the system's safety and/or liveness.
To address these problems, we design a two-round committee formation mechanism.
In the first round, participants who have mined the latest blocks are automatically elected as a committee.
This rules out participants without contributions, and thus Sybil attacks.
In the second round, each participant in the committee computes a verifiable random function (VRF)~\cite{micali1999verifiable, vrf} and is elected into a sub-committee if its random output meets an election difficulty parameter.
As computing VRF requires the knowledge of the secret key, a participant only learns the election outcome of itself, but not others, in the second round.
This requires no interaction between participants and makes the probability of predicting committee members negligible.
}

\begin{figure}[t]
    \centering
    \includegraphics[width=2.2in]{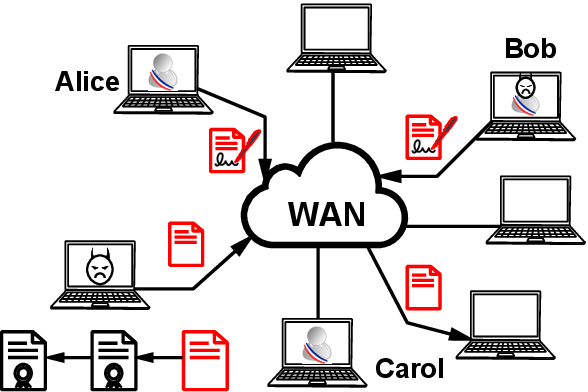}
    \caption{A showcase of QC in Crystal. The attacker corrupts Bob and has to obtain a signature from either Alice or Carol to form a QC of its block.}
    \vspace{-2mm}
    \label{fig: architecture}
\end{figure}


Crystal exhibits significant improvement in resisting withholding-related attacks.
For selfish mining, compared to Bitcoin where a miner with $40\%$ of the computation power can obtain about $52.6\%$ of block rewards~\cite{eyal2014majority}, this miner cannot gain more profits than honest mining in Crystal.
For double-spending attacks, compared to Bitcoin where a miner with $30\%$ of the computation power reverts a transaction with $6$ confirmations with a probability of $15.6\%$~\cite{nakamoto2012bitcoin,rosenfeld2014analysis}, the success probability for this attacker falls below $0.62\%$ in Crystal.
In addition, as Crystal makes simple and minimal changes over NC, Crystal can provide formal end-to-end safety proof ensuring that no unknown attacks are introduced.
The simplicity also allows Crystal to be a drop-in replacement for NC, making NC-style protocols resist withholding-based attacks.


\vspace{1mm} \noindent \textbf{Contributions.} We summarize our contributions as follows:
\begin{itemize}[leftmargin=*]
    \item We present Crystal, which, to our knowledge, is the first NC-style protocol that resists block withholding.

    \item We propose a two-round committee formation protocol 
    in the permissionless setting and also provide a reward mechanism to incentivize participants to follow the protocol.

    \item We provide a formal analysis of the benefits gained by introducing quorum certificates in the mining process, including the impact on selfish mining and double-spending attack.

    \item We analyze the protocol overheads. The results show that the overhead (i.e., storage and computation) of Crystal is negligible.

    \item We discuss how to generalize the usage of QC and also show how to apply the idea of Crystal to other Proof-of-X (PoX) blockchain protocols.
\end{itemize}

\vspace{1mm} \noindent \textbf{Roadmap.}
The rest of this paper is organized as follows.
We provide backgrounds for our work in Section~\ref{sec:background}.
We then introduce the system model and goals in Section~\ref{sec: system model}.
We present our Crystal protocol in Section~\ref{sec: protocol} and formally analyze its security in Section~\ref{sec: analysis}.
We evaluate Crystal in Section~\ref{sec: evaluation}.
We then briefly describe how to extend and generalize Crystal in Section~\ref{sec:disscussion}.
We provide the related work in Section~\ref{sec:related work} and conclude the paper in Section~\ref{sec:conclusion}.

\section{Background} \label{sec:background}
\subsection{Bitcoin and Nakamoto Consensus} \label{subsec:nc}
Bitcoin is the first permissionless blockchain that realizes a distributed ledger.
In Bitcoin, transactions are packaged into blocks.
Each block consists of a block header and a set of transactions.
The block header includes a hash value of the previous block (also called parent block), a timestamp, a Merkle root of transactions, a nonce (\jn{i.e., a random string~\cite{nakamoto2012bitcoin}}) and other metadata.
Blocks are linked together by hash references to form the chain structure.

Bitcoin relies on Nakamoto Consensus (NC) to maintain a consistent blockchain among participants. NC contains two components: the Proof-of-Work (PoW) and the longest chain rule. 
Specifically, each participant follows PoW to generate new blocks, in which a block contains a valid nonce such that the hash value of the entire block header is below a certain
threshold $D$ depending on the difficulty level—a system parameter that can be adjusted.
This puzzle-solving process is often referred to as \emph{mining} and so the participants are referred to as \emph{miners}.
\jn{The difficulty of PoW determines the chance of finding a new block on each try, and by adjusting the difficulty, the block generation rate can be controlled. For example, blocks are mined every 10 minutes on average in Bitcoin.}

Once a new block is produced, it ideally should reach all participants before the next block is produced.
If this happens to every block, then each participant in the system will observe the same chain of blocks.
In reality, multiple blocks may be produced around the same time, and so they may share a common parent block, leading to a fork.
To resolve such a fork, an honest miner always accepts the longest chain as the valid one and mines on it.
This is known as the longest chain rule.
Each block included in the longest chain can bring its miner a block reward as well as some fees for including transactions (also called transaction fees). 
These incentives encourage miners to devote their computational resources to the system.
Besides, each participant \emph{confirms} blocks (and their transactions) in the longest chain except for the last $k-1$ blocks, where $k$ is determined by the security level. For example, $k = 6$ in Bitcoin. This is also referred to as the $k$-deep confirmation rule.

\subsection{Attacks with Block Withholding}
\label{subsec:attacks}
Researchers have proposed several attacks against NC; all of them involve a key operation called block withholding. 
These attacks can be broadly categorized into two types: incentive-based and safety-based attacks. An attacker performs an incentive-based attack to gain more profits than they could by following the protocol.
In a safety-based attack, the attacker targets the safety of the system (e.g., reverting a block and the including transactions confirmed by participants).
A representative is the double-spending attack~\cite{nakamoto2012bitcoin}.
Sometimes, an attacker can combine these two types of attacks to obtain more profits (or to worsen safety) than that from a single attack~\cite{sapirshtein2016optimal,gervais2016security}.

\vspace{1mm}
\noindent \textbf{Incentive-based Attacks.}
NC is designed to fairly reward miners according to their contributions to the system (i.e., miners' revenue should be proportional to their devoted computation power).
However, the studies in~\cite{eyal2014majority, sapirshtein2016optimal, nayak2016stubborn, gervais2016security} show that malicious miners can gain more revenue than their fair share by deviating from the protocol.
This mining attack is called selfish mining.
In this attack, an attacker can withhold its newly generated blocks and then strategically publish these private blocks to invalidate honest blocks, thereby obtaining higher block rewards than it deserves.

\vspace{1mm}
\noindent \textbf{Safety-based Attacks.}
There are two influential safety-based attacks. The first one is the double-spending attack~\cite{nakamoto2012bitcoin}. In a successful attack case, an attacker first creates a payment using its digital token (e.g., Bitcoin) to the merchant for some service or goods. Once the transaction is confirmed in a public chain, and the service or goods are delivered, the attacker will invalidate the public chain by publishing its longer private chain, which contains a conflicting transaction to transfer the token back to itself.
In this way, a successful double-spending attack ruins the safety of the system.
Another attack is targeted at consistency, which is first proposed by Pass et al.~\cite{Pass2017}. In this attack, an attacker simply delays all honest blocks as long as possible to slow the growth rate of the honest chain.
Meanwhile, the attacker mines secretly on its private chain and strategically publishes private blocks to maintain a forking chain to prevent honest participants to have a consistent blockchain.

\section{System Model, Goals and Notations} \label{sec: system model}
In this section, we present the system model, system goals, and required cryptographic notations. In particular, we list used notations in Table~\ref{table:Notation}.

\begin{table}[t]
    \centering
    \caption{\textbf{Summary of Notations.}}
    \begin{tabular}{@{}ll@{}}
        \toprule[1pt]
        Notation                                             & Description                             \\
        \midrule
        $\mathcal{N}$                                        & Set of miners                           \\
        $n$                                                  & Number of miners                        \\
        $m$                                                  & \# of sub-committee members                 \\
        $f$                                                  & \# of Byzantine committee members       \\
        $W$                                                  & Sliding window size                     \\
        \jn{$D$}                                                  & \jn{Mining target}                       \\
        \jn{$d$}                                                & \jn{Committee election target}                      \\
        $QC$                                                 & Quorum certificate                      \\
        $\lambda$                                            & Mining process rate                     \\
        $\mathcal{M}_B$                                      & Block $B$'s committee                   \\
        $\Delta$                                             & Message delivery time bound             \\
        \textsf{H($\cdot$)}                                  & Cryptographic hash function             \\
        $\beta$                                              & Fraction of honest computation power    \\
        $\alpha$                                             & Fraction of Byzantine computation power \\
        $(pk_i ,sk_i)$                                       & Key pair for committee construction     \\
        $(\tilde{pk_i}, \tilde{sk_i})$                       & Key pair for signing messages           \\
        $\mathsf{Sign}(\cdot)/\mathsf{Verify}(\cdot)$        & Signature generation/verification       \\
        $\mathsf{VRFProve}(\cdot)/\mathsf{VRFVerify}(\cdot)$ & VRF generation/verification             \\
        \bottomrule[1pt]
    \end{tabular}
    \label{table:Notation}
\end{table}

\subsection{System Model} \label{subsec: system-model}
We consider a set of $n$ participants, denoted by $\mathcal{N}$.
Each participant $i \in \mathcal{N}$ can generate two public/secret key pairs:
a pair $(pk_i ,sk_i)$ of a unique signature scheme~\cite{uniqueKey} for committee formation, and
a pair $(\tilde{pk_i}, \tilde{sk_i})$ of a regular digital signature scheme for signing and verifying messages.
There is no trusted public key infrastructure.\
\jn{
The participants in the set $\mathcal{N}$ can be divided into three types: altruistic, rational, and Byzantine, following the BAR model~\cite{BAR2015}.
Specifically, altruistic participants (also called honest participants in blockchains~\cite{garay2015bitcoin, Pass2017}) always follow the protocol.
Honest participants may represent enthusiasts who support the system or have faith in long-term returns if the system operates smoothly.
Therefore, honest participants are conservative on Byzantine behaviors to break the system or short-term deviations from the protocol for more profits, both of which will affect the long-term benefits.
Rational participants are profit-driven and aim at maximizing their expected utility from participation.
Rational participants may represent external speculators who may deviate from the protocol for short-term profits.
Byzantine participants can behave arbitrarily from the protocol for any reason.
All Byzantine participants are assumed to be controlled by a single attacker.}

\noindent \textbf{Network Model.}
We consider a permissionless setting in which participants may join or leave the system at any time.
\jn{Following NC-style protocols~\cite{garay2015bitcoin,Pass2017}, we assume the network is synchronous.
That is, whenever an honest participant transmits a message (e.g., a block, a transaction, or a signature), it takes up to $\Delta$ seconds for the message to reach all other honest participants.}
This assumption is acceptable since it is only used for simplifying analysis, but not for our protocol design.

\noindent \textbf{Mining Model.}
The block mining process of all miners can be further modeled as a Poison process with the rate $\lambda$~\cite{bagaria2019deconstructing,Race2020}.
Each participant $i$ has a limited amount of computation power, measured by the number of hash computations that it can run per second.
The computation power controlled by honest participants (resp. Byzantine and rational participants) is $\beta$ (resp.  $\alpha$) fraction of the total computation power.
Clearly, we have $\alpha + \beta = 1$.
At any time, honest participants are assumed to control the majority of computation power, i.e., $\beta > 1/2$.
Let $\lambda_h$ (resp. $\lambda_a$) denote the block mining rate of honest miners (resp. Byzantine and rational participants).
We have $\lambda_h = \beta \lambda$ and $\lambda_a = \alpha \lambda$.
Besides, as $\beta > 1/2$, we have $\lambda_h > \lambda_a$.
Blocks mined by honest participants are referred to as honest blocks.

\subsection{Design Goals}

\jn{
    We aim at building a permissionless blockchain protocol that satisfies the following properties with an overwhelming probability under the system model in Section~\ref{subsec: system-model}:
    \begin{itemize} [leftmargin=*]
        \item \emph{Safety:} No two honest participants commit two different blocks at the same height.
        \item \emph{Liveness:} If an honest participant receives a valid transaction, then the transaction will eventually be included in all honest participants' committed blocks.
        \item \emph{Incentive compatibility:} any subset of colluding rational participants cannot gain more revenue by deviating from the protocol.
    \end{itemize}
    The safety property also corresponds to the consistency property~\cite{Pass2017, Kiffer2018}. The safety and liveness properties should be held in the presence of Byzantine participants, while the incentive compatibility property should be guaranteed with rational participants.}

\subsection{Cryptographic Notations}
The system employs standard cryptographic techniques, including the digital signature and the verifiable random functions (VRF)~\cite{micali1999verifiable, vrf}.
\jn{The collision-resistant cryptographic hash function $\mathsf{H}(msg)$ outputs a hash value given the input message $msg$}.
The digital signature scheme contains two functions:
\jn{$\mathsf{Sign}(sk, msg)$ that takes secret key $sk$ and message $msg$ as input and outputs signature $\sigma$ on $msg$; and
$\mathsf{Verify}(pk, \sigma, msg)$ that takes public key $pk$, signature $\sigma$ and message $msg$ as input and outputs $1$ if the $\sigma$ is valid and $0$ if not.}
The VRF scheme contains two functions:
$\mathsf{VRFProve}(x, sk)$ that takes secret key $sk$ and input $x$ as input and outputs a pseudorandom output $y$ and a proof $\pi$; and
$\mathsf{VRFVerify}(x, y, \pi, pk)$ that takes input $x$, output $y$, proof $\pi$ and public key $pk$ and outputs $1$ if the $y$ is valid and $0$ if not. 

\section{Crystal Design} \label{sec: protocol}
In this section, we provide the design of Crystal.
We start from a strawman design that is intuitive but fails to address some challenges, then improve the strawman design step-by-step to address these challenges, yielding Crystal.

\subsection{Strawman Design}\label{subsec:strawman}
\jn{We first provide a strawman design that adds mining transparency to NC.
The setting of the protocol is the same as Bitcoin except that there is a fixed committee of $m$ members. At most $f$ ($f \leq \left \lfloor \frac{m}{2} \right \rfloor)$ of committee members are Byzantine.
Besides, public keys of the committee members are known, e.g., are recorded on the blockchain.} 
On receiving a valid block $B$, a committee member will sign the block $B$ and broadcast its signature to all participants.
\jn{When participants obtain at least $f+1$ distinct signatures on $B$, they combine signatures into a QC, and start mining the next block $B^\prime$ whose block header contains $B$'s block hash and QC.}
The rest of the mining process is the same as NC.
\jn{Here, including the QC in the block header is the key to achieving mining transparency.
To obtain sufficient signatures, the block proposer has to publish the block to at least $f+1$ committee members, where at least one committee members are honest and will propagate the block to other participants.}

However, challenges arise to realize such a protocol based on NC.
First, the committee needs to be carefully chosen in a permissionless setting that is subject to Sybil attack~\cite{sybil2002}.
Second, the committee is an easy target for attacks (e.g., bribery, Eclipse attacks~\cite{eclipse2015}, and DoS attack~\cite{chen2021efficient}) if the identity is public.
Third, it lacks an incentive mechanism to encourage committee members to behave correctly.
Crystal extends the strawman design to address the above three challenges.
Specifically, Crystal addresses the first two challenges by employing a novel two-round committee formation protocol in Section~\ref{subsec: committee} and addresses the last challenge by employing a reward mechanism for committee members in Section~\ref{subsec: reward}.




\subsection{Data Structure}

\noindent \textbf{Block.}
Transactions issued by clients are batched into blocks.
A block contains a block header, a set of transactions, and a quorum certificate $QC$.
A block header inherits most of the fields from the Bitcoin header and also contains several new fields.
Specifically, a block header consists of the following fields:

\begin{itemize} [leftmargin=*]
    \item \textbf{$\mathcal{\tau}$}: the root of the Merkle tree~\cite{Merkle} containing a set of transactions,
    \item \jn{\textbf{$h$}: is the hash of the parent block},
    \item \textbf{$\mathcal{\tau}_{QC}$}: the root of the Merkle tree containing all certificate proofs $\mathcal{K}$ of the parent block,
    \item \textbf{$\tilde{pk_i}$}: the block owner's public key for receiving rewards,
    \item \textbf{$pk$}: the block owner's public key generated by the unique signature scheme for committee formation,
    \item \jn{\textbf{$D$}: is the value determining the difficulty of finding new blocks,}
    \item \jn{\textbf{$e$}: is a Unix timestamp,}
    \item \textbf{$\eta$}: is a valid nonce, i.e., a random string so that the hash value of $B$'s header is less than the target $T$.
\end{itemize}

\noindent \textbf{Blocktree.}
Each participant locally maintains a collection of blocks, which are either mined by itself or received from other participants.
\jn{These blocks are linked by hash references, forming a rooted tree called \emph{blocktree}.}
A block's height is the shortest path length between it and the genesis block $\mathcal{G}_0$.
The height of the genesis block is zero.

\noindent \textbf{Quorum Certificate.}
Each block in a valid blocktree determines a unique committee, which will be described later in Section~\ref{subsec: committee}.
Each member of the committee could have multiple membership shares.
\niu{In Crystal, the expected committee size is $m$, and a quorum certificate has to contain votes from members that have at least $\lfloor m/2 \rfloor$ + 1 memberships.}
A block is certified if associated with a valid QC.

\subsection{Mining Process}
\jn{Algorithm~\ref{alg: mining} presents the pseudo-code of the mining process in Crystal. Participants use PoW to mine blocks and accept blocks by the longest chain rule. Specifically, before mining a block, a participant first chooses a set of transactions \textit{Txs} and prepares a block header \textit{hdr} (Line 4-5). Then, the participant tries to solve a PoW puzzle by constantly changing the nonce field until a valid nonce is found (Line 5-7). Here, the valid nonce ensures that the hash value of the header is smaller than the target $D$. After that, the participant will create a new block (containing the block header \textit{hdr}, transaction set \textit{Txs}, and other metadata), process the block (introduced shortly), and broadcast the block to the network (Line 8-10).}

\jn{
When receiving a new block $B$, a participant has to process the block (Line 11-19). The participant first checks the validity of the block, which contains four requirements: $1$) $B$ should contain a valid nonce; $2)$ $B$'s parent block should be locally available; $3$) $B$'s transactions should be valid; and $4$) $B$ should contain a valid $QC$ for its parent block. 
If all requirements are met, then the participant will append the block $B$ to its local blocktree.
After that, the participant will update the last block in the longest chain (Line 18) and check whether it should vote for the block (Line 19). Note that if there are two longest chain branches, miners follow the uniform tie-breaking rule, i.e., randomly choosing the last block at one of them. The voting function will be introduced shortly.}

\vspace{1mm} \noindent \jn{\textbf{Mining difficulty adjustment.} Crystal follows a similar mining difficulty adjustment algorithm (DAA) as Bitcoin~\cite{nakamoto2012bitcoin} to maintain the average block interval at about 10 minutes. 
In particular, the difficulty is periodically adjusted, i.e., every 2016 blocks (i.e., roughly two weeks), by deterministically changing the target value $D$ by the following equation:
\[D_{new} = D_{old} \times \frac{\text{Time of last 2016 blocks}}{2016 \times \text{10 minutes}},\]
where block time is computed by the Unix timestamp included in the block header.}

\begin{algorithm}[t]
    \small
    \caption{\jn{The block mining process}} \label{alg: mining}
    \begin{algorithmic}[1]
        \Statex \textbf{Local State}:
        \State $M \leftarrow \{\mathcal{G}_0 \}$ \quad \quad \quad $\rhd${ the set of blocks}
        \State $B^{\prime} \leftarrow \mathcal{G}_0$ \quad \quad \quad \quad $\rhd$ { the last block in the longest chain}
        \Statex
        \State \textbf{function} $\textsf{MineBlock}()$
        \State\hspace{\algorithmicindent}
        $\textit{Txs} \leftarrow$ \textsf{getTransactions}()
        \State\hspace{\algorithmicindent}
        $\textit{hdr} \leftarrow$ \textsf{createHeader}($B^{\prime}$, \textit{Txs})
        \State\hspace{\algorithmicindent}
        \textbf{While} $\textsf{H}(\textit{hdr}) > D$ \textbf{do} 
        \State \hspace{\algorithmicindent} \hspace{\algorithmicindent}
        $\textit{hdr.nonce} \leftarrow$ \textsf{getNewNonce}()
        \State \hspace{\algorithmicindent} 
         $B \leftarrow$  \textsf{createBlock(\textit{hdr},\textit{Txs})};
        \State \hspace{\algorithmicindent}
        \textsf{ProcessBlock}($B$)
        \State \hspace{\algorithmicindent}
        \textsf{Broadcast(B)}
        \Statex
        \State \textbf{function} \textsf{ProcessBlock}($B$)
        \State\hspace{\algorithmicindent}
        verify that \textsf{H}($B$.\textit{hdr}) $<$ D
        \State\hspace{\algorithmicindent}
        verify that $\textit{B.hash} = \textsf{H}(A)$ for block $A \in M$
        \State\hspace{\algorithmicindent}
        verify that $\textit{B.Txs}$
        \State\hspace{\algorithmicindent}
        verify that $\textit{B..\textit{hdr}.$\mathcal{\tau}_{QC}$}$
        \State\hspace{\algorithmicindent}
        if (any above verification fails) \textbf{then return}
        \State\hspace{\algorithmicindent}
        $M \leftarrow M \cup \{B \}$
        \State \hspace{\algorithmicindent}
        $B^{\prime} \leftarrow \textsf{getLastBlockinLongestChain}(M)$
        \State \hspace{\algorithmicindent}
        \textsf{voteForBlock}($B$) 
    \end{algorithmic}
\end{algorithm}

\subsection{Committee Formation} \label{subsec: committee}
We first provide an overview of the committee formation and then present the detailed design. 

\vspace{1mm} \noindent \textbf{Overview.}
In the strawman design, assuming the existence of a fixed committee is unrealistic in the permissionless setting of the blockchain.
To address the issues, we propose a two-round committee formation protocol, as shown in Figure~\ref{fig: overview}.
First, the latest $W$ blocks form a sliding window.
Each block in this sliding window can bring its owner a membership share, and participants with membership shares form a committee.
Whenever a new block is appended to the chain, the sliding window moves one block forward, and the share of the old block beyond the sliding window expires.
Since the committee members are chosen because of their dedicated resources, this approach is resilient to Sybil attacks.
Additionally, the sliding window guarantees that the members of the committee are recently active participants who have a higher chance to stay online.
This is referred to as the first-round election.

With the committee, a second-round election is used to construct the final subcommittee for each block in a random and non-interactive way.
Informally, the secondary election works as sortition for each block, which chooses a random subset of block owners in the sliding window according to their block shares (i.e., the number of blocks in the sliding window belonging to the same participant).
The randomness input for each sortition is the hash value of the new block.
As participants cannot predict the output of the cryptographic hash function, the resulting input is unpredictable.
\niu{Note that when the length of the blockchain is less than the sliding window size $W$, Crystal runs as the vanilla Nakamoto Consensus until the bootstrapping phase completes.} 

\vspace{1mm} \noindent \textbf{Detailed Design.}
Algorithm~\ref{alg: election} provides the specification of the two-round committee election protocol.
Here, participants check whether to vote for blocks that have passed verification rules (see Algorithm~\ref{alg: mining}).

\begin{figure}[t]
    \centering
    \includegraphics[width=3.5in]{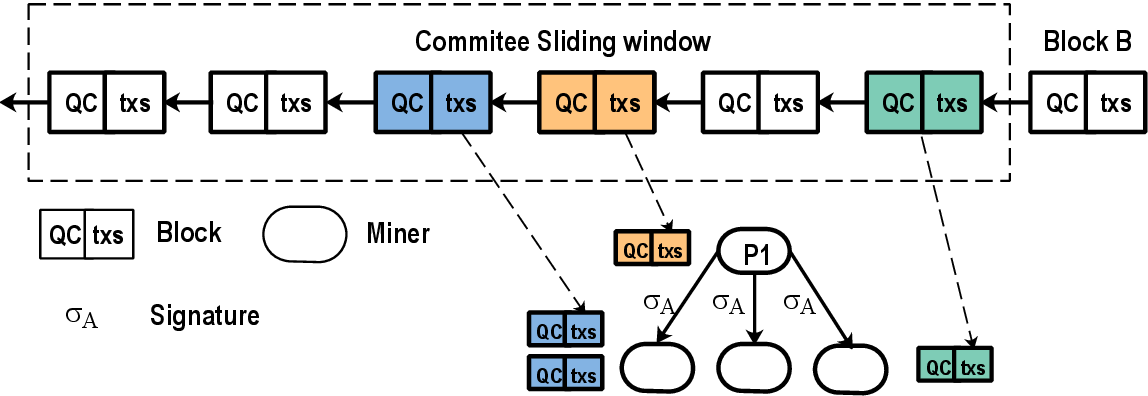}
    \caption{\textbf{Overview of the Crystal design.} The colored blocks are randomly selected in the recent block window and the owners of the colored blocks form the committee for block $B$. Upon receiving block $B$, participant $P1$ knows that it is in the committee, then it will create and broadcast a signature for block $B$.}
    \vspace{-2mm}
    \label{fig: overview}
\end{figure}

Upon receiving a new block $B$, the participant $i$ triggers the function $\mathsf{voteForBlock}(B)$, which determines if it is eligible to vote for the block $B$.
\niu{Note that a participant can vote for multiple blocks at the same height.}
Specifically, participant $i$ iterates over the sliding window of the last $W$ blocks.
Note that to bootstrap, Crystal runs as the vanilla Nakamoto Consensus for the first $W$ blocks.
In every iteration, if the block is mined by participant $i$ itself, it computes $\textsf{VRFProve}(B.hash || index, sk)$ (Line 8-9).
Here, $B.hash$ is the hash of block $B$'s header, $index$ is the position index of $B^{\prime}$ in the sliding window, and $sk_i$ is participant $i$'s secret key.
If the result is smaller than the election target $d$ (Line 10), 
the participant $i$ creates a vote $eCert$ containing $B.hash$, its public key $pk_i$, the block index $index$, the VRF output $y$ and proof $\pi$ (Line 11), adds $eCert$ to the vote set $QC_B$, and broadcasts $eCert$ to other participants (Line 13).
\jn{Here, $d$ determines the probability that a block in the sliding window brings its owner a membership share and is set to guarantee that the \emph{expected} number of membership shares of all participants is $m$. 
For example, if the pseudorandom output of VRFProve() function is uniformly distributed on $[0, 2^{256}]$, then $d = m/W \times 2^{256}.$ 
The parameters $W$ and $d$ are fixed and hard-coded in the system. Section~\ref{subsec: parameter} shows how to determine the value of these parameters.}

Upon a new vote $eCert$, a participant triggers the function $\mathsf{onReceiptVote}(eCert)$, which verifies $eCert$ and assembles the QC. 
If $eCert$ passes the validity check (Line 21, 24-25) and the generator of $eCert$ is chosen by the two-round election (Line 22), then $eCert$ is added to the vote set $QC_B$ on $B$.
If the vote set $QC_B$ contains more than $m/2$ valid votes on $B$, then these votes constitute a valid QC of $B$ (Line 27-28).

\begin{algorithm}[t]
    \caption{Crystal's two-round committee election protocol from a participant $i$'s perspective.}\label{alg: election}
    \begin{algorithmic}[1]
        \Statex
        \textbf{Parameter}:
        \State $d$ \jn{\quad \quad \quad \quad$\rhd$ Election difficulty}
        \Statex
        \textbf{Local State}:
        \State $(pk_i, sk_i)$ \jn{\quad $\rhd$ The $i$-th participant's key pair}
        \Statex
        \Statex \jn{$\rhd$\quad Triggered upon every new block $B$}
        \State \textbf{function} \textsf{voteForBlock}($B$):
        \State\hspace{\algorithmicindent}
        $QC_B \gets \{\}$
        \State\hspace{\algorithmicindent}
        $B^{\prime} \leftarrow \textsf{getParentBlock}(B)$
        \State\hspace{\algorithmicindent}
        \textbf{for} $j \in [1, W]$
        \State\hspace{\algorithmicindent}\hspace{\algorithmicindent}
        \textbf{if} $B^{\prime}.pk = pk_i$ \textbf{then}
        \State\hspace{\algorithmicindent}\hspace{\algorithmicindent}\hspace{\algorithmicindent}
        $index \leftarrow j$
        \State\hspace{\algorithmicindent}\hspace{\algorithmicindent}\hspace{\algorithmicindent}
        $ (y, \pi) \leftarrow  \textsf{VRFProve}(B.hash || index, sk)$
        \State\hspace{\algorithmicindent}\hspace{\algorithmicindent}\hspace{\algorithmicindent}
        \textbf{if} $y < d$ \textbf{then}
        \State\hspace{\algorithmicindent}\hspace{\algorithmicindent}\hspace{\algorithmicindent}\hspace{\algorithmicindent} 
        $eCert\leftarrow \left \langle pk_i, B.hash, index, y, \pi \right \rangle $
        \State\hspace{\algorithmicindent}\hspace{\algorithmicindent} \hspace{\algorithmicindent}\hspace{\algorithmicindent}
        $QC_B \gets QC_B \cup eCert$
        \State\hspace{\algorithmicindent}\hspace{\algorithmicindent} \hspace{\algorithmicindent}\hspace{\algorithmicindent}
        \textsf{Broadcast($eCert$)}
        \State\hspace{\algorithmicindent}\hspace{\algorithmicindent}
        $B^{\prime} \leftarrow \textsf{getParentBlock}(B^\prime)$
        \Statex
        \Statex \jn{$\rhd$ Triggered upon every new vote $eCert$}
        \State \textbf{function} \textsf{onReceiptVote}($eCert$):
        \State\hspace{\algorithmicindent}
        $index \leftarrow eCert.index$
        \State\hspace{\algorithmicindent}
        $hash \leftarrow eCert.hash$
        \State\hspace{\algorithmicindent}
        $y \leftarrow eCert.y$
        \State\hspace{\algorithmicindent}
        $pk \leftarrow eCert.pk$ \
        \State\hspace{\algorithmicindent}
        $\pi \leftarrow eCert.\pi$
        \State\hspace{\algorithmicindent}
        \textbf{if} $\nexists B | \textsf{H}(B) = hash$ \textbf{then return}
        \State\hspace{\algorithmicindent}
        \textbf{if} $y < d \wedge 1 \leq index \leq W $ \textbf{then}
        \State\hspace{\algorithmicindent} \hspace{\algorithmicindent}
        $B^{\prime} \leftarrow \textsf{findBlockByIndex}(B, index)$
        \State\hspace{\algorithmicindent} \hspace{\algorithmicindent}
        \textbf{if} $pk = B^{\prime}.pk$ \textbf{then}
        \State\hspace{\algorithmicindent} \hspace{\algorithmicindent} \hspace{\algorithmicindent}
        \textbf{if} $\mathsf{VRFVerify}(hash || index, y, \pi, pk)$ \textbf{then}
        \State \hspace{\algorithmicindent}\hspace{\algorithmicindent}\hspace{\algorithmicindent}\hspace{\algorithmicindent}
        $QC_B \gets QC_B \cup eCert$
        \State\hspace{\algorithmicindent}
        \textbf{if} $\left| QC_B \right| > m/2$ \textbf{then}
        \State\hspace{\algorithmicindent}\hspace{\algorithmicindent}
        \textsf{$QC_B$ becomes a valid QC}
    \end{algorithmic}
\end{algorithm}

\vspace{1mm} \noindent \textbf{Comparison with Existing Committee Construction.}
Existing works~\cite{ByzCoin,pass2017hybrid,elastico} leverage a one-round method to form the committee.
Specifically, the participants that have blocks in the sliding share window form the committee for all subsequent blocks.
Despite the simplicity, this committee formation approach has two issues.
First, if an attacker controls the majority of membership shares, it can simply refuse to sign for any new honest blocks extending the blockchain.
As a result, no honest blocks can be certified and further extended.
In such a case, Byzantine participants will fully control the system.
Second, the identities of the committee members are known to the public, so they are subject to DoS or Eclipse attacks.
Once the attacker controls or isolates participants with the majority of the membership shares, it can also control or stall the system as explained before. To counter this issue, the simplest way is to increase the size of the sliding window so that the probability that the majority of membership shares are controlled by the attacker (or members controlling the majority of membership shares are under attack) is overwhelmingly small.
However, the large window size leads to a large overhead of QC formation.

By contrast, the two-round committee formation protocol enables the system to have better security and efficiency.
First, the second round is non-interactive, meaning that participants can independently determine if they are chosen for the committee (Line $11$).
The elected committee members are unknown to the attacker.
This allows us to choose a larger sliding window size and a smaller committee size, making it hard to attack the committee of a target block.
Second, as the inputs (i.e., blocks' hash values) of the sortition algorithm are different, different blocks will have different committees.
This implies that the system is more robust to the failure of committee formation, and so allows a higher failure probability.
This further means that the committee size could be smaller, leading to better efficiency.

\subsection{Reward Mechanism} \label{subsec: reward}

Crystal demands a reward mechanism that incentivizes participants to 1) participate in Crystal, and 2) follow the protocol as specified.
Participants' actions in Crystal can be classified into two categories: \emph{mining-related} and \emph{voting-related}.
Rewarding mining-related actions encourages rational participants to participate in the mining process.
Rewarding voting-related actions encourages rational participants who have blocks in the sliding window to stay online.
In addition, the reward mechanism needs to incentivize committee members to respond quickly, since participants would start mining the next block once they collect a quorum number of votes.
Moreover, the reward mechanism needs to incentivize participants to include more signatures as a QC.

Algorithm~\ref{alg: reward} provides the specification of Crystal's reward mechanism.
Apart from common reward mechanisms in permissionless blockchains (i.e., block rewards and transaction fees in Line 1), Crystal provides the following additional rewards.
First, a participant who has produced a block $B$ in the main chain additionally receives an inclusion reward $R_i$ for each vote included in $B$ (Line 5).
Second, if a participant is elected into the committee of $B$ and has voted for $B$, it receives a voter reward $R_v$ multiplied by the number of its votes (Line 3-4).
Here, a participant uses its public key of a regular digital signature scheme to receive rewards.

Note that participants may have different membership shares for different forking blocks, and so can obtain different voting rewards.
Although a rational participant would prefer the block that brings it higher rewards to be included in the blockchain, it would vote for all of them.
This is because blocks' committees are different, and members in the unflavored blocks' committees would vote for these blocks.
If a rational participant does not vote for some forking blocks that it is entitled to, it would suffer from economic loss once these unflavoured blocks win.
In addition, Crystal can adopt suitable parameters to guarantee that a subset of rational participants' votes does not affect the acceptance of blocks.

\begin{algorithm}[t]
    \caption{Crystal's reward mechanism.}\label{alg: reward}
    \begin{algorithmic}[1]
        \Statex
        \textbf{function} \textsf{getRewards}($B$)
        \State \hspace{\algorithmicindent}
        \textcolor{gray}{Pay block reward and tx fees in $B$ to $B.\tilde{pk}$}
        \State \hspace{\algorithmicindent}
        \textbf{for} each $eCert$ in $B.QC$
        \State \hspace{\algorithmicindent} \hspace{\algorithmicindent}
        $\tilde{pk} \leftarrow \textsf{findRewardPK}(eCert.pk)$
        \State \hspace{\algorithmicindent} \hspace{\algorithmicindent}
        Pay $R_v$ to $\tilde{pk}$ \jn{\quad \quad $\rhd$ vote reward}
        \State \hspace{\algorithmicindent} \hspace{\algorithmicindent}
        Pay $R_i$ to $B.\tilde{pk}$ \jn{ \quad $\rhd$ vote inclusion reward}
    \end{algorithmic}
\end{algorithm}

\section{Analysis} \label{sec: analysis}
In this section, we first analyze the security of committees.
Then, we show that Crystal can effectively resist block withholding misbehavior.
Finally, we study how Crystal thwarts the selfish mining attack and the double-spending attack and guarantees safety against all attacks.

\subsection{Security of Committees} \label{subsec: committee security}
In Crystal, each block is accompanied by a committee, which is randomly constructed from participants that have blocks in the sliding window. Besides, Crystal should guarantee the following good committee property of blocks.

\jn{\begin{definition}[Good Committee] For a given block $B$ and its committee $\mathcal{M}_B$, we say that $\mathcal{M}_B$ has good committee property if
    the honest membership shares are greater than $m/2$ and the membership shares controlled by the attacker are no greater than $m/2$.
\end{definition}}

The good committee property implies that honest participants control the majority of membership shares and 
Also, this property implies the following two conditions:
$i)$ if block $B$ is a malicious block, the attacker cannot have a QC for this block without receiving signatures from some honest
member in $\mathcal{M}_B$; $ii)$ if block $B$ is an honest block, all honest participants can have a QC for this block even if the attacker does not sign $B$.
It turns out that by choosing suitable parameters, $\mathcal{M}_B$ has good committee property with high probability.

\begin{lemma} \label{lem: majority}
    Given $0 < \epsilon < 1$, for any sufficiently large $m$, $\mathcal{M}_B$ has good committee property except for probability at most $\epsilon$.
\end{lemma}

\begin{proof}
    The committee construction can be modeled as a random sampling problem with two possible independent outcomes: selected and not-selected.
    The probability for each block's owner to be selected into the committee is $p = m / W$, where $m$ is the expected sub-committee size, and $W$ is the sliding window size (See Section~\ref{subsec: committee}).
    Note that if $W$ is large enough, the probability of picking an honest (or dishonest) participant is $\beta$ (or $\alpha$).
    Let $Y$ (resp, $X$) denote the random variable corresponding to the number of honest (resp, dishonest) membership shares. Let $A$ and $B$ denote the events that $\{Y: 0 \leq Y \leq \left \lfloor m/2 \right \rfloor\}$ and $\{X: \left \lceil m/2 \right \rceil \leq X \leq W\}$, respectively.
    The probability $\Pr \left [A \cup B \right ]$ that the elected committee does not satisfy the good committee property is then
    \begin{equation}
        \begin{split}
            \Pr \left [A \cup B \right ] &= \Pr \left [A \right ] + \Pr \left [B \right ] -  \Pr \left [A \cap B \right ] \\
            & \leq \Pr \left [A \right ] + \Pr \left [B \right ] \\
            & = \sum_{j=\left \lceil m/2 \right \rceil}^{W}{\Pr \left [ X = j \right ]} +  \sum_{j=0}^{\left \lfloor m/2 \right \rfloor}{\Pr \left [ Y = j \right ]}.
        \end{split}
        \label{eq:failure1}
    \end{equation}
    Since $Y$ follows a Binomial distribution, we have
    \begin{equation}
        \Pr \left [ Y = j \right ] = {\binom{W}{j} (\beta p)^{j}(1- \beta p)^{W-j}},
        \label{eq:failure3}
    \end{equation}
    where $\beta p$ is the probability that an honest membership share is produced.
    Similarly, we have
    \begin{equation}
        \Pr \left [ X = j \right ] =
        {\binom{W}{j} (\alpha p)^{j}(1- \alpha p)^{W-j}}.
        \label{eq:failure2}
    \end{equation}

    By using standard Chernoff bounds, we can show that the failure probability $\Pr \left [A \cup B \right ]$ decreases exponentially in $m$.
    This implies that $\Pr \left [A \cup B \right ]$ can be made smaller than $\epsilon$ for sufficiently large $m$.
\end{proof}
We will apply this lemma to study the impact of the failure probability $\epsilon$ on block withholding behaviors in Section~\ref{subsec: withheld} as well as to choose suitable  $\epsilon$ and subcommittee size $m$ in Section~\ref{subsec: parameter}.

\begin{table}[t]
    \centering
    \caption{The probability $P_l$ of successfully withholding $l$ blocks in Crystal. Let $T_{f}$ denote the average failure time.}
    \begin{tabular}{@{}ccccccccc@{}}
        \toprule[1pt]
        \multirow{2}{*}{} & \multicolumn{2}{c}{$\epsilon = 10^{-2}$} &           & \multicolumn{2}{c}{$\epsilon = 10^{-3}$} &           & \multicolumn{2}{c}{$\epsilon = 10^{-4}$}                            \\
        \cmidrule{2-3}
        \cmidrule{5-6}
        \cmidrule{8-9}
                          & l =2                                     & l =3      &                                          & l =2      & l =3                                     &  & l =2      & l = 3     \\
        \toprule[1pt]
        $P_l$             & $10^{-2}$                                & $10^{-4}$ &                                          & $10^{-3}$ & $10^{-6}$                                &  & $10^{-4}$ & $10^{-8}$ \\

        $T_{f}$           & $16.6$h                                  & $9.9$w    &                                          & $6.9$d    & $1.9$y                                   &  & $9.9$w    & $190.2$y  \\
        \bottomrule[0.9pt]
    \end{tabular}
    \label{table:withhold}
\end{table}

\subsection{Resisting Block Withholding} \label{subsec: withheld}
The main goal of Crystal is to resist block withholding behaviors with high probability.
Informally, when an attacker finds a block, it has to include a QC of this block to mine the next block.
A valid QC has to include the signatures of committee members who control the majority of memberships.
If the good committee property of a block holds, the attacker cannot create a QC for this block by itself and has to publish the block to some honest committee members to obtain their signatures.
Otherwise, the attacker can withhold a certified block and continue to mine the next block.
Let $N_{A}$ denote the blocks secretly withheld by the attacker.
Since the probability that a block's committee does not have the good committee property is $\epsilon$, the probability $P_l$ that the attacker can withhold $l$ consecutive blocks is:
\[P_l = \Pr \left [ N_{A} = l \right ] = \epsilon^l (1-\epsilon).\]

Table~\ref{table:withhold} displays the probability $P_l$ of successfully withholding $l$ blocks in Crystal.
To give a concrete example, let us assume that the failure probability is $10^{-4}$ (which is the setting used in Crystal)
and committees operate independently of each other. Since blocks in Crystal are mined every ten minutes on average, it would take $9.9$ weeks to have two blocks whose committee do not have the good committee property, and an attacker who controls $30\%$ of computation power would take about $33$ weeks on expected to successfully withhold two blocks (since the expected block interval for the attackers' blocks is $33.3$ minutes.).
This example shows that block withholding behaviors can be effectively resisted with high probability in Crystal.
With this, we will next show that Crystal can thwart the selfish mining and double-spending attack.

\subsection{Selfish Mining Attack} \label{subsec:incentive}
In this section, we demonstrate quantitatively how Crystal thwarts the most famous incentive-based attack: selfish mining \cite{eyal2014majority, sapirshtein2016optimal,nayak2016stubborn}.
In particular, we prove that any subset of colluding rational participants cannot gain more from deviating from the protocol when the majority of computation power is controlled by honest participants.

Before analysis, We first provide a simple example, which can show the benefits of eliminating block withholding in defending the selfish mining attack.
Figure~\ref{fig: selfish} illustrates two simple cases of the selfish mining attack in NC and Crystal.
Suppose at the starting time, both the attacker and honest participants have the same blockchain (hidden at the left of the vertical dashed line).
The case in NC is already explained in Section~\ref{sec: intro}, in which the attacker publishes two withheld blocks to invalidate one honest block. In this way, the attacker not only gets two block rewards by the longest chain rule but also makes the honest block lose one block reward.
Similarly, in Crystal, the attacker first mines block $A_1$ and withholds it in order to create a fork.
Since the attacker cannot get a QC for block $A_1$ (when the good committee property holds),
it cannot mine the next block, thereby losing the opportunity of mining more blocks (e.g., the red dashed block $A_2$ in Figure~\ref{fig: selfish}) and winning the associated block rewards.
Later, the attacker publishes block $A_1$ to match an honest block $H_1$ (and a fork is created).
As the two branches have the same length, the fork cannot be resolved until the next block is mined.
By contrast, in NC, the attacker may publish two withheld blocks and has a higher chance of winning the fork race.
This clearly explains why it is more difficult for selfish mining to be profitable in Crystal.

\begin{figure}[t]
    \centering
    \includegraphics[width=2.5in]{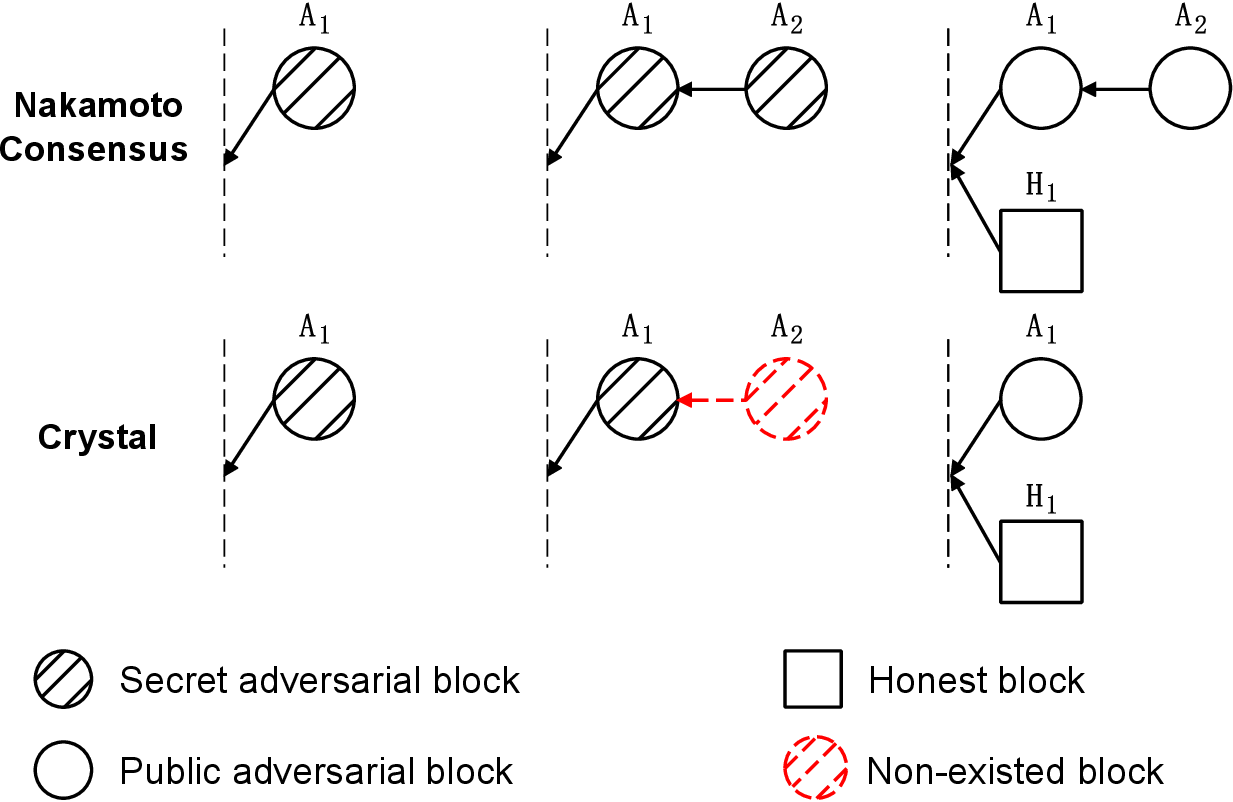}
    \caption{\textbf{A simple comparison between selfish mining in NC and Crystal.} The red dashed circle block represents the lost opportunity of generating a new block in Crystal.}
    \vspace{-2mm}
    \label{fig: selfish}
\end{figure}

We now turn to formal analysis. For simplicity, we assume that the delay $\Delta$ is close to zero. Thus, no blocks are mined when other blocks are transmitted.
This assumption is widely used in the analysis of selfish mining for Bitcoin~\cite{eyal2014majority, sapirshtein2016optimal,nayak2016stubborn}~\footnote{The assumption can be justified because the block generation interval in Bitcoin (i.e., $10$ minutes on average) is significantly larger than the block propagation delay (e.g., smaller than $10$ seconds \cite{Decker2013}). Crystal adopts a similar setting.}.
We will later relax this assumption and consider the impact of the network delay in our evaluations.
We present the Markov model of selfish mining attack on Crystal in Figure~\ref{fig: mining}.
For ease of presentation, we re-scale the time axis so that the attacker generates new blocks at rate $\alpha$, and honest participants (as a whole) generate new blocks at rate $\beta$.
Besides, we introduce the same parameter $\gamma$ as in~\cite{eyal2014majority, sapirshtein2016optimal,nayak2016stubborn}, which denotes the fraction of honest participants that are mining on blocks produced by the attacker (rather than by the honest participants) whenever they observe two forking branches of equal length. The value of $\gamma$ depends on the fork choice rule. In Crystal, $\gamma = 1/2$ because Crystal adopts the uniform tie-breaking rule (see Section~\ref{sec: protocol}).
We are now ready to introduce the states for our Markov model:
\begin{itemize} [leftmargin=*]
    \item $S_0$: honest participants and the attacker have the same view of the blockchain;

    \item $S_1$: the attacker's chain is one block longer than the public chain;

    \item $S_{0^{\prime}}$: the attacker's chain has an equal length with the honest participants', and honest participants randomly choose one of the branches to mine the next block.
\end{itemize}

By solving the above Markov model, we have the following theorem of selfish mining in Crystal.

\begin{theorem} \label{theorem:selfish}
    For $\alpha \in [0, 0.5)$, the attacker in Crystal cannot gain a higher fraction of blocks included in the longest chain than $\alpha$ in the long term.
\end{theorem}

The theorem shows that the attacker cannot obtain a higher fraction of blocks through selfish mining. As we said previously, there are two main reasons. First, the attacker has to stop mining for launching the selfish mining attack, resulting in a loss of mining new blocks. (See the red dashed line in Figure~\ref{fig: mining}.)
Second, the decreased block leads make the attacker hard to win the fork competition. Due to space limits, detailed proof of this theorem is provided in Appendix~\ref{appen: selfish}.
With this theorem, we can easily derive the fraction of block rewards won by the attacker and have the following corollary.

\begin{figure}[t]
    \centering
    \includegraphics[width=2in]{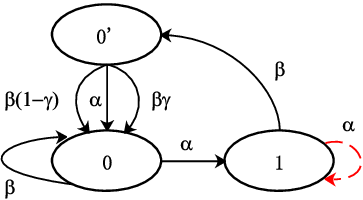}
    \caption{\textbf{The Markov process of the selfish mining in Crystal.} The red dashed line represents the nonexistent transition that the attacker mines the next block on its withheld block, and is drawn to simplify the analysis.}
    \vspace{-2mm}
    \label{fig: mining}
\end{figure}

\begin{corollary}
    For $\alpha \in [0, 0.5)$, the attacker in Crystal cannot gain a higher fraction of block rewards from the selfish mining attack than honest mining in the long term.
\end{corollary}

\begin{proof}
    Each block eventually ended in the longest chain can bring its participant a fixed block reward, and so the fraction of blocks in the main chain equals the fraction of block rewards won by participants (see Section~\ref{subsec: reward}).
\end{proof}
The corollary shows that the attacker cannot obtain a higher fraction of block rewards from the selfish mining attack. In the NC-style blockchains (e.g., Bitcoin and Ethereum), block rewards are currently significantly higher than the transaction fees to incentivize participants' participation. This also holds in Crystal. Hence, the transaction fees are usually ignored in the analysis of the selfish mining~\cite{eyal2014majority, sapirshtein2016optimal,nayak2016stubborn}.
Therefore, we can conclude that the attacker cannot gain a higher revenue by the selfish mining attack in Crystal.

\begin{remark}
    The above analysis ignores the failure probability of the good committee property. Indeed, the failure probability (i.e., $10^{-4}$) is too small to affect the results.
    This also holds for the analysis of the double-spending attack.
\end{remark}

\subsection{Safety} \label{subsec:safety}
In this section, we first illustrate how Crystal improves the security against the double-spending attack~\cite{nakamoto2012bitcoin, rosenfeld2014analysis, ozisik2017explanation}.
Then, we prove that Crystal can provide safety no matter what the attacker does.

\subsubsection{Double-Spending attack}\label{subsec:spending}
To better understand how Crystal thwarts the double-spending attack, we first provide a simple example, which illustrates the attack in NC and Crystal.
Recall that participants in NC and Crystal follow the $k$-deep confirmation rule, i.e.,  a block as well as its containing transactions are confirmed if this block is included in the longest chain and extended by at least $k-1$ other blocks.
Hence, the goal of the double-spending attack is to revert a transaction ``confirmed'' by some honest participants.

\begin{figure}[t]
    \centering
    \includegraphics[width=3.2in]{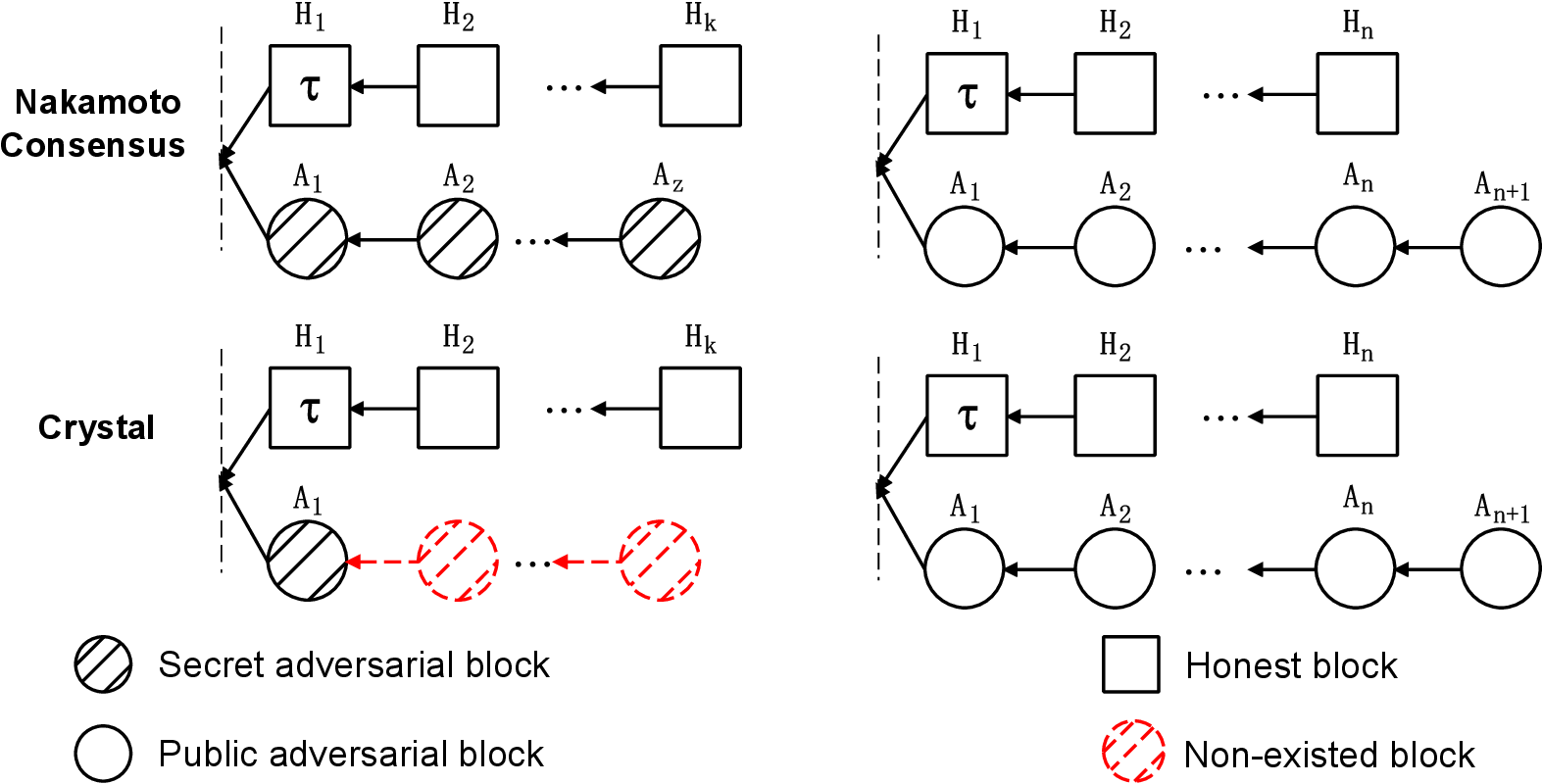}
    \caption{The $k$-confirmation rule in Crystal. When the merchant receives block $H_k$, it releases its goods. Then, the attacker tries to catch up and eventually has a longer forking branch than the public branch containing block $H_1$.
        For simplicity, we do not draw the quorum certificate links.}
    \vspace{-2mm}
    \label{fig: double}
\end{figure}

Let us first look at the double-spending attack in NC, as shown in Figure~\ref{fig: double}. The double-spending attack process can be divided into two phases~\cite{rosenfeld2014analysis, ozisik2017explanation}. First, the attacker buys some off-chain goods and pays the merchant with a transaction $\tau$.
Then, the transaction is included in a block $H_1$, which is further extended by some subsequent blocks $H_2, H_3$, . . . , known to all.
Meanwhile, the attacker secretly builds another forking branch $A_1, A_2, A_3,...$, and block $A_1$ contains a conflicting transaction $\tau^{\prime}$, which double-spends the coin containing in the transaction $\tau$.
In the second phase, the merchant releases the goods once receiving block $H_k$ by the $k$-deep confirmation rule. Then, the attacker tries to publish a longer branch than the public one.
If succeeds, all participants will accept transaction $\tau^{\prime}$ by the longest chain rule, and the merchant will neither have goods nor be paid.

In the double-spending attack, the attacker can mine as many blocks as possible, but cannot release these blocks before the $k$-confirmation is obtained.
This is because the publication of the conflicting transaction $\tau^{\prime}$ will alert the merchant, possibly resulting in transaction cancellation.
Hence, in Crystal, the attacker can only withhold at most one block and cannot mine the next block before the merchant receives block $H_k$
due to the good committee property.
The second phase is similar to that in NC, i.e., the attacker tries to publish a longer branch.
Intuitively, due to the decrease of withheld blocks in the first phase, the success probability for the double-spending attack will decrease significantly as $k$ increases.

Next, we compute the success probability for the double-spending attack. We assume that the message delay $\Delta$ approaches zero by following the prior analysis in~\cite{nakamoto2012bitcoin, rosenfeld2014analysis, ozisik2017explanation}. (See Section~\ref{subsec:incentive} for more detailed reasons.)
We will remove the limitation in our later evaluations.
Following a similar analysis of the double-spending attack that developed in~\cite{nakamoto2012bitcoin, rosenfeld2014analysis, ozisik2017explanation}, we can obtain the following theorem.
\begin{theorem}
    For $\alpha \in [0, 0.5)$, the success probability for the double-spending attack is $(\frac{\alpha}{1-\alpha})^{k-1}$.
\end{theorem}

\begin{proof}
    We assume that the attacker can pre-mine one block before the attack by following the analysis in ~\cite{rosenfeld2014analysis,ozisik2017explanation}.
    Therefore, at the beginning of the second phase, the attacker is $k - 1$ blocks behind the public chain.
    The probability the attacker will ever catch up, and eventually have at least one block longer than the public branch is $(\frac{\alpha}{\beta})^{k-1} = (\frac{\alpha}{1-\alpha})^{k-1}$ by the analysis in~\cite{rosenfeld2014analysis}.
\end{proof}
By this theorem, we can compute the suitable $k$ given the attacker's computation power $\alpha$ and the allowed successful probability of the double-spending attack.
Recall that $\alpha$ is less than $0.5$, so $\frac{\alpha}{1-\alpha} < 1$, and the successful probability is exponentially decreasing as $k$ increases.
For better flow, we defer the results and the comparisons with the results of the double-spending attack in NC to Section~\ref{sec: evaluation}.

\subsubsection{Safety Analysis}
In this section, we prove that Crystal can maintain safety against \emph{all} possible attacks. The aforementioned double-spending attack is only one specific safety attack.
To better understand the detailed analysis, we first provide a proof sketch.

\vspace{1mm} \noindent \textbf{Proof Sketch.} Our safety analysis of Crystal follows the existing methods of NC~\cite{garay2015bitcoin, Pass2017,Kiffer2018,Jing2020,Jing2021}.
\jn{In particular, safety analysis of NC relies on a special kind of block called \emph{converged block} (aka convergence opportunity~\cite{Pass2017,Kiffer2018} or double-lagger~\cite{Jing2021}).} 
A converged block is an honest block that is mined at time $t$, and in the previous and next $\Delta$ period there are no other honest blocks that are mined.
The previous $\Delta$ guarantees that the block is mined after one of the longest chains.
The next $\Delta$ guarantees that all honest participants will receive the block and mine after it by the longest chain rule.
Furthermore, if there are no adversarial blocks to match the block, all honest participants will accept the chain ending with this block, reaching the safety of the same chain.
Informally speaking, the safety analysis of NC is to show that during a period, with high probability, there always exists such converged blocks no matter what the attacker does.

Similarly, In Crystal, if a block $D$ is mined at time $t$, and in the previous and next $2\Delta$ period there are no other blocks are mined, the block will be the uniquely certified block at its height.
Similarly, the $2\Delta$ period can guarantee that the block is the only certified honest block at its height.
If there is no adversarial block to match the block, the chain that ends with this block will be accepted by all honest participants.
Therefore, the analysis goal is also to show that during a period, with high probability, there always exist such \emph{converged} blocks no matter what the attacker does.
As the only difference between the safety analysis of NC and Crystal lies in the conditions for converged blocks (i.e., $\Delta$ vs. $2 \Delta$), the existing safety proofs of NC~\cite{garay2015bitcoin, Pass2017, Kiffer2018, Jing2020, Jing2021} can be easily extended to Crystal.

\vspace{1mm} \noindent \textbf{Analysis.} We first define converged blocks and then prove that Crystal has the safety property. We provide several important results here, while deferring the detailed proofs to Appendix~\ref{appen: safety proof}.

\begin{definition}[Converged blocks]
    A block produced by honest miners at time $t$ is called a converged block if no honest block is produced in the previous and next $2 \Delta$ time.
\end{definition}

\begin{lemma}\label{lem:delta}
    If an honest block $B_{\ell}$ is mined at time $t_0$, then every honest participant can start to mine a block of at least height $\ell+1$ by the time $t_0 + 2\Delta$.
\end{lemma}
This lemma shows that the maximum delay for honest participants to mine on top of a new honest block in Crystal is at most $2\Delta$.
In reality, members in committee $\mathcal{M}_{B_{\ell}}$ can sign block $B_{\ell}$ once they receive $B_{\ell}$.
Due to space limits, the formal proof is provided in Appendix~\ref{appen: safety proof}.

Next, we will prove the safety property, i.e., honest participants do not confirm different blocks at the same height of blockchains.

\begin{lemma}
    Suppose that block $B$ is a converged block of height $\ell$, then $B$ is the only honest block of height $\ell$.
\end{lemma}

This lemma says that every converged block has a unique height among all other honest blocks. The detailed proof is provided in Appendix~\ref{appen: safety proof}.
Next, we give the bounds on the number of converged blocks in a time interval of $t$.

\begin{lemma} \label{lemma:converged}
    Let $\eta = e^{-2 \beta f \Delta }$. For any $0 < \delta <1$, the number of converged blocks mined in a time interval $t$ is at least $(1+\delta) \eta^2  \beta f t $, except for $e^{-\Omega\left(\delta^2 \eta^2 \beta f t \right)}$ probability.
\end{lemma}

This lemma says that in a period of $t$, the number of converged blocks is at least $(1+\delta) \eta^2  \beta f t $. The detailed proof is provided in Appendix~\ref{appen: safety proof}.
If at one height there is only one converged block, then all miners will reach a consensus on the same chain of blocks. Therefore, to ruin the safety, the goal of the attacker is to match every converged block with a malicious block. Finally, we can give the safety theorem as follows.

\begin{theorem}[Safety]\label{thm:safety}
    Suppose $\eta^2 \beta > (1+\delta) \alpha$.
    If $B$ and $B'$ are two distinct blocks of the same height, then they cannot be
    both confirmed, each by an honest participant. This property holds, regardless of malicious action, except for $e^{-\Omega\left(\delta^2 \min\{ \eta^2 \beta, \alpha \} k \right)}$ probability.
\end{theorem}

This theorem states that if $\eta^2 \beta > (1+\delta) \alpha$, there are no conflicting blocks confirmed by honest participants at the same height no matter what the attacker does.
This shows that the introduced quorum certificate does not ruin the safety of NC.
The detailed proof is provided in Appendix~\ref{appen: safety proof}.

\begin{figure}[t]
     \centering
     \includegraphics[width=2.4in]{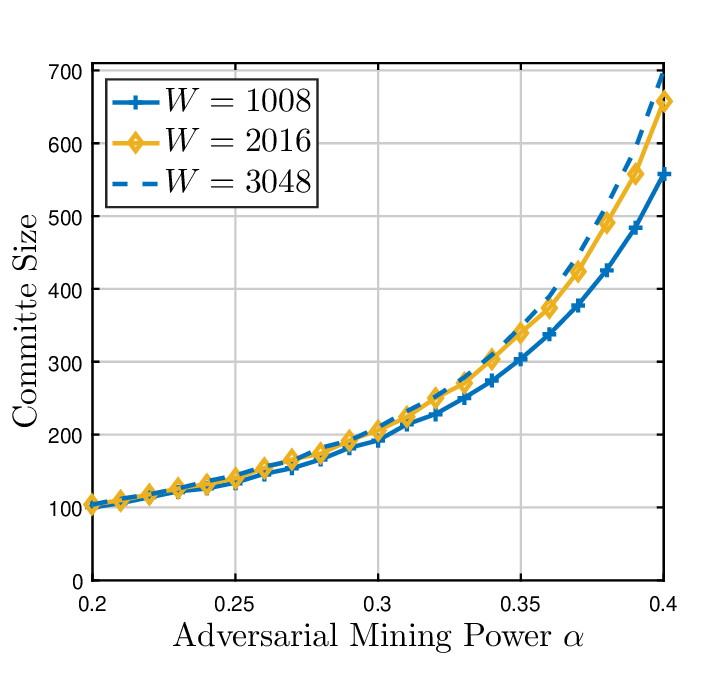}
    \caption{The sub-committee size $m$ that is sufficient to limit the failure probability of a committee to $10^{-4}$ with different malicious computation power $\alpha$ and sliding window size $W$. 
       \label{fig:paramters}
    }
\end{figure}

\section{Evaluation} \label{sec: evaluation}
In this section, we first evaluate the protocol overhead in terms of storage and network costs and then evaluate Crystal's resistance to selfish mining attacks and double-spending attacks under different parameter settings.

\subsection{Parameter Choice} \label{subsec: parameter}
Crystal has three system parameters: the sliding window size $W$, the expected subcommittee size $m$, and the maximum fraction $\alpha$ of computation power controlled by the attacker.
According to Section~\ref{subsec: committee security}, the three parameters have to be chosen such that the good committee property fails with a probability smaller than $10^{-4}$, which suffices to resist block withholding behaviors.
\jn{
Since the sub-committee is elected from members in the sliding window, the sliding window size $W$ restricts the sub-committee size $m$ and is first decided. 
Figure~\ref{fig:paramters} shows the required sub-committee size $m$ with different malicious computation power $\alpha$ and sliding window size $W$. The results show that the increase in the sliding windowing size will slightly lead to an increase in $m$. 
In Crystal, the sliding window size $W$ is $3024$ to prevent an attacker from corrupting all block owners. 
In particular, $W = 3024$ means that the committee is elected from the participants who have mined blocks in the last $3$ weeks, for an average block mining rate of $10$ minutes. Note that Crystal runs as the vanilla Nakamoto Consensus for the first $W$ blocks as a bootstrapping. This is because Lemma~\ref{lem: majority} cannot hold when the sliding window size $W$ is small.
}

Next, we plot the expected subcommittee size $m$ that is needed with different $\alpha$, for a violation probability $10^{-4}$ of the good committee property in Figure~\ref{fig:paramters}.
\jn{On one hand, we would like the system can work even with a higher fraction of Byzantine computation power $\alpha$.
On the other hand, we would like $m$ to be small such that the bandwidth costs of quorum certificates are low.}
The figure clearly shows the trade-off: the stronger the assumption on the fraction of computation power held by the attacker ($\alpha$), the larger the committee size needs to be.
To balance these two targets, we choose $\alpha = 35\%$ and $m$ is larger than $340$.
Particularly, due to the impact of offline participants (See Section~\ref{subsec:offline}), we set $m$ to $500$.

\begin{figure}[t]
    \centering
    \setlength{\abovecaptionskip}{10pt}   
    \setlength{\belowcaptionskip}{10pt}   
    \includegraphics[width=0.65\linewidth]{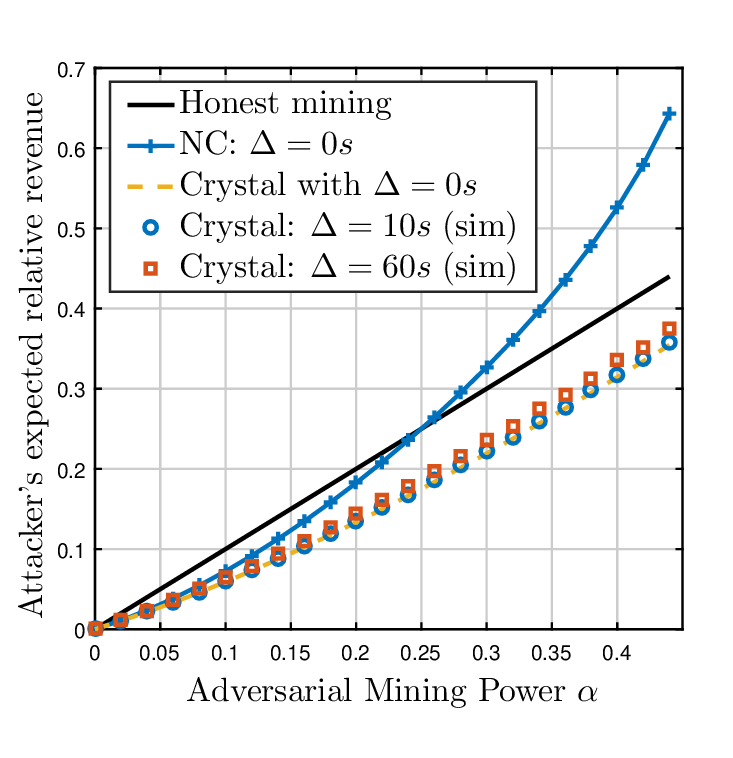}
    \caption{Relative revenue of the selfish participant in Crystal and Bitcoin.}
    \vspace{-2mm}
    \label{fig: selfishmining}
\end{figure}

\subsection{Protocol Overhead}
\jn{In this subsection, we evaluate the introduced overhead of Crystal in terms of communication and storage costs.
Specifically, each block has one additional quorum certificate $QC$, which contains a set of membership proofs (i.e., $eCert$s).
Since the membership proofs contain redundant fields (the hash of parent block and public key of signers), which can be extracted from previous blocks. 
Therefore, we only keep the VRF's output and proof in each $eCert$ with an additional bit map to denote which block is elected.
In Crystal, each VRF's output is $32$ Bytes, and the proof is $64$ Bytes\footnote{\url{https://github.com/w3f/schnorrkel/blob/master/src/vrf.rs}}.
Since the expected sub-committee size $m$ is $500$, in the worst case where each committee member has one membership share and all committee members sign for a block, the additional storage is about $48$ KB per block.
Considering the average 1MB block size in Bitcoin, a QC constitutes around $4.8\%$ of today’s blocks.
However, the top ten mining pools in Bitcoin control more than $90\%$ of computation power~\cite{bitcoinscan}.
This implies that the number of participants in the committee will be much smaller than the committee size.
Besides, the committee size for each block varies, and the probability that a block in the sliding window is selected is $p = 500/3024$.
The probability of having a committee with a size of more than $700$ is
\begin{equation}
    {\sum_{j=700}^{3024}\binom{3024}{j} p^{j}(1- p)^{3024-j}} \approx 1.33\times10^{-12},
    \label{eq:offline}
\end{equation}
which is considered low in practice.
Such a committee brings an additional overhead of about $67.2$ KB, constituting roughly $6.7\%$ of a $1$MB block.}

\begin{table*}[ht]
    \centering
    \caption{The probability of a successful double-spending attack in NC (black) and Crystal (red) under different fraction $\alpha$ of malicious computation power and number $k$ of confirmation blocks.}
    \begin{tabular}{@{}c c c c c | c c c c@{}}
        \toprule[1pt]
        \multicolumn{1}{c}{}    & \multicolumn{4}{c}{$\Delta = 0$} & \multicolumn{4}{c}{$\Delta = 10$}           \\
        \hline
        $\alpha|k$              & $2$                              & $4$                               & $6$                              & $8$                              & $2$ & $4$ & $6$ & $8$ \\
        \hline
        \multirow{2}{*}{$0.1$}  & $5.60\times10^{-2}$              & $5.46\times10^{-3}$               & $5.91\times10^{-4}$              & $6.73\times10^{-5}$
                                & $5.72\times10^{-2}$              & $5.66\times10^{-3}$               & $6.14\times10^{-4}$              & $7.18\times10^{-5}$                                      \\
                                & \color{red}{$1.23\times10^{-2}$} & \color{red}{$1.52\times10^{-4}$}  & \color{red}{$1.88\times10^{-6}$} & \color{red}{$2.32\times10^{-8}$}

                                & \color{red}{$1.28\times10^{-2}$} & \color{red}{$1.64\times10^{-4}$}  & \color{red}{$2.30\times10^{-6}$} & \color{red}{$2.40\times10^{-8}$}                         \\
        \hline
        \multirow{2}{*}{$0.2$}  & $2.08\times10^{-1}$              & $6.67\times10^{-2}$               & $2.33\times10^{-2}$              & $8.48\times10^{-3}$
                                & $2.21\times10^{-1}$              & $6.86\times10^{-2}$               & $2.44\times10^{-2}$              & $8.81\times10^{-3}$                                      \\
                                & \color{red}{$6.46\times10^{-2}$} & \color{red}{$4.01\times10^{-3}$}  & \color{red}{$2.44\times10^{-4}$} & \color{red}{$1.53\times10^{-5}$}

                                & \color{red}{$6.25\times10^{-2}$} & \color{red}{$3.91\times10^{-3}$}  & \color{red}{$2.65\times10^{-4}$} & \color{red}{$1.64\times10^{-5}$}
        \\
        \hline
        \multirow{2}{*}{$0.3$}  & $4.32\times10^{-1}$              & $2.52\times10^{-1}$               & $1.56\times10^{-1}$              & $1.00\times10^{-1}$
                                & $4.35\times10^{-1}$              & $2.58\times10^{-1}$               & $1.59\times10^{-1}$              & $1.01\times10^{-1}$                                      \\
                                & \color{red}{$1.90\times10^{-1}$} & \color{red}{$3.37\times10^{-2}$}  & \color{red}{$6.20\times10^{-3}$} & \color{red}{$1.14\times10^{-3}$}

                                & \color{red}{$1.84\times10^{-1}$} & \color{red}{$3.51\times10^{-2}$}  & \color{red}{$6.61\times10^{-3}$} & \color{red}{$1.19\times10^{-3}$}
        \\
        \hline
        \multirow{2}{*}{$0.4$}  & $7.04\times10^{-1}$              & $5.80\times10^{-1}$               & $4.93\times10^{-1}$              & $4.26\times10^{-1}$
                                & $7.16\times10^{-1}$              & $5.88\times10^{-1}$               & $5.03\times10^{-1}$              & $4.38\times10^{-1}$                                      \\
                                & \color{red}{$4.44\times10^{-1}$} & \color{red}{$1.98\times10^{-1}$}  & \color{red}{$8.78\times10^{-2}$} & \color{red}{$3.90\times10^{-2}$}

                                & \color{red}{$4.61\times10^{-1}$} & \color{red}{$2.08\times10^{-1}$}  & \color{red}{$9.57\times10^{-2}$} & \color{red}{$4.38\times10^{-2}$}                         \\
        \hline
        \multirow{2}{*}{$0.45$} & $8.51\times10^{-1}$              & $7.83\times10^{-1}$               & $7.34\times10^{-1}$              & $6.93\times10^{-1}$
                                & $8.60\times10^{-1}$              & $7.93\times10^{-1}$               & $7.45\times10^{-1}$              & $7.06\times10^{-1}$                                      \\
                                & \color{red}{$6.69\times10^{-1}$} & \color{red}{$4.78\times10^{-1}$}  & \color{red}{$3.00\times10^{-1}$} & \color{red}{$2.01\times10^{-1}$}

                                & \color{red}{$6.93\times10^{-1}$} & \color{red}{$4.48\times10^{-1}$}  & \color{red}{$3.30\times10^{-1}$} & \color{red}{$2.28\times10^{-1}$}                         \\
        \bottomrule[0.9pt]
    \end{tabular}
    \label{table:withholding-nc-crystal}
\end{table*}

\subsection{Selfish Mining and Double-spending attacks}
In this section, we run Monte Carlo simulations to show how Crystal thwarts selfish mining and double-spending attacks.
We consider two settings of the maximum block propagation delay: $\Delta = 0$ seconds and $\Delta = 10$ seconds (The measured block propagation delay in Bitcoin~\cite{Decker2013}).
In particular, if an honest block of height $l$ is first mined at time $t$, all the honest blocks found before $t + \Delta$ (resp, $t + 2\Delta$) are considered to be a forking block of the same height $l$ in NC (resp. Crystal).
In this way, we can consider the impact of block propagation delay on the attacks.

\subsubsection{Selfish Mining} Figure~\ref{fig: selfishmining} shows the proportion of block rewards won by attackers in Crystal and Bitcoin with different malicious computation power.
\jn{First, the results show that when considering the message delay (10s and 60s), the attacker can win a higher fraction of block rewards in Crystal. But the increase is not significant so it is reasonable for us to ignore the delay in the analysis.}
Second, the fraction of rewards won by the attacker in Crystal is lower than that of honest mining, no matter what the malicious computation power $\alpha$ is.
By contrast, in Bitcoin, when $\alpha \geq 0.25$, the attacker can gain more revenue through the selfish mining attack.
Therefore, the results show that Crystal can significantly thwart the selfish mining attack, even considering the network delay.

\subsubsection{Double-spending attack.}
Table~\ref{table:withholding-nc-crystal} shows the probability of a successful double-spending attack in NC and Crystal for different $\alpha$ and confirmed block $k$. \niu{Here, we use Equation~(1) in~\cite{rosenfeld2014analysis} to compute the associated probability of the double-spending attack.}
Let us first look at the results when $\Delta = 0$.
The results show that given some $k$ and $\alpha$, the probability of success for the attack in Crystal is significantly smaller than that in Bitcoin.
For example, when $k$ is large and $\alpha$ is small (e.g., $k = 8$ and $\alpha \leq 0.3$), the probability of success for the double-spending attack in Crystal can decrease by almost two orders of magnitude compared with that in Bitcoin.
On the other hand, given the same probability of success and malicious computation power $\alpha$, participants in Crystal can confirm a transaction faster.
For example, when $\alpha = 40\%$ and $k = 8$, an attacker can succeed with $42.6\%$ probability, whereas in Crystal, the attacker can only succeed with the almost same probability of success when $k = 2$. Therefore, in this case, participants can reduce the transaction latency from $80$ minutes in Bitcoin to $20$ minutes in Crystal.

When considering the network delay, both the probability of success for the attack in Crystal and Bitcoin are a little higher than the results with $\Delta = 0$.
This is because, honest participants may not contribute their computation power to the longest chain, reducing the security against the attacker.
However, the results show that Crystal still can thwart double-spending attacks compared with NC, even in the worst case that the delay for collecting a QC is assumed to be as long as the block propagation delay.

\section{Discussion} \label{sec:disscussion}
In this section, we first discuss the impact of offline participants. Then, we discuss how Crystal can be extended in a more general form to use QC and the extension to other PoX blockchain protocols.

\subsection{Offline Participants} \label{subsec:offline}
Crystal is a permissionless blockchain, participants can join or leave the system anytime.
If a participant whose blocks are in the sliding window leaves the network, it may not vote for subsequently mined blocks when it is elected into blocks' committees.
As a result, the failure probabilities of the good committee property of blocks' committees will increase.
Furthermore, some honest blocks cannot be certified with a QC and so be extended.

Crystal leverages two methods to address the issues caused by offline participants.
First, Crystal provides incentives for participants to stay online.
Specifically, these participants can receive voting rewards if they are elected into blocks' committees and also vote for them (See Section~\ref{subsec: reward}).
These rewards are set much higher than the cost of staying online and participating in protocols.
Hence, this can prevent a large fraction of participants who have the potential to vote from staying offline.
Second, Crystal adopts suitable parameters to make the impact caused by offline block owners small.
We choose a larger sub-committee size $m$ than needed.
Figure~\ref{fig: offline} shows the failure probabilities of certifying an honest block with different ratios of offline participants.
Here, we assume that Byzantine participants always stay online and that they do not vote for honest blocks at all.
The figure shows that even though about $10\%$ of participants leave the network, the failure probability is still below $10^{-3}$.
However, on today's Internet, it is infeasible to keep so many participants to stay offline.

\begin{figure}[t]
    \centering
    \setlength{\abovecaptionskip}{10pt}   
    \setlength{\belowcaptionskip}{10pt}   
    \includegraphics[width=0.65\linewidth]{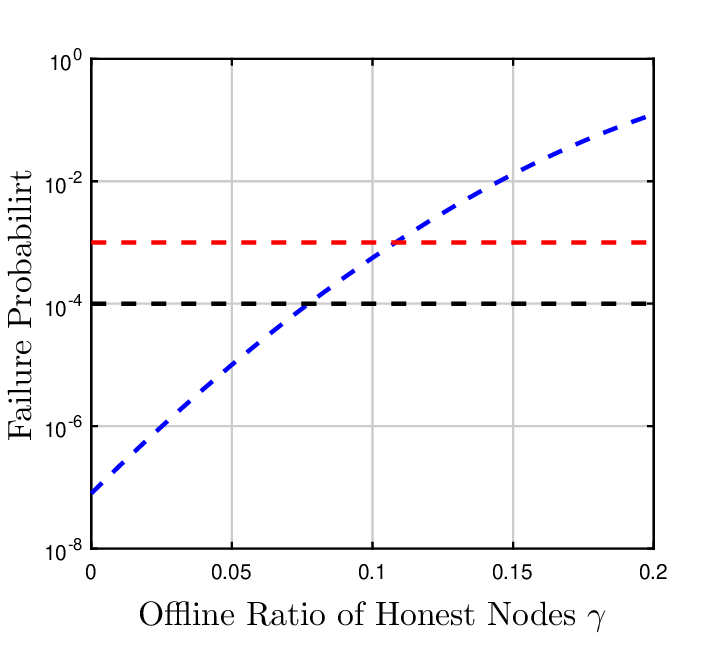}
    \caption{The failure probability $p$ of certifying honest blocks with different ratios $\gamma$ of offline participants.}
    \vspace{-2mm}
    \label{fig: offline}
\end{figure}

\subsection{Quorum Certificate Distance}
In Crystal, each block has to contain a QC referencing its parent block (see Section~\ref{sec: protocol}).
However, this usage of the QC in Crystal is just one special case and can be extended to more general ways.
For clarity, we first introduce the concept of certificate distance, which is the height difference between a block and the block that its quorum certificate is referencing.
\jn{
Figure~\ref{fig: framework} illustrates two use cases of certificate distance one (used in Crystal) and two. 
Suppose a miner has mined block B in both cases. 
In the first case, to mine the next block C, the miner has to publish block B and waits an additional period for forming the QC of block B.  
By contrast, in the second case, the miner may directly produce block C without publishing block B since the QC of block A could be formed before block B is mined. 
From the above cases, we can see that a larger certificate distance will allow an attacker to withhold more blocks (i.e., a decrease in mining transparency), but leads to a shorter time for forming QCs.}
In a nutshell, by controlling the certificate distance, we can make a trade-off between the mining transparency and the delay of waiting for QCs.

\begin{figure}[t]
    \centering
    \includegraphics[width=3.3in]{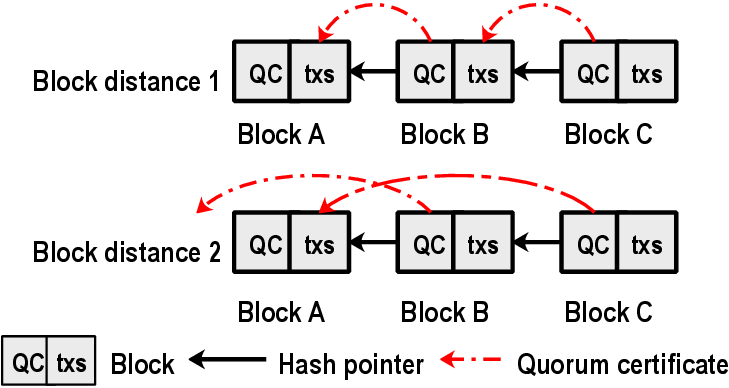}
    \caption{\textbf{A simple illustration of certificate distance one and two.}}
    \vspace{-2mm}
    \label{fig: framework}
\end{figure}

\subsection{Proof-of-X}
In PoW, miners produce blocks by solving mathematical puzzles, which require massive computation power and consume a huge amount of electricity.
\jn{To make NC more efficient, many alternatives such as Proof-of-Stake (PoS)~\cite{bentov2016snow, epos, Gilad2017}, Proof-of-Authority (PoA)~\cite{Parity}, as well as Proof-of-Elapsed-Time (PoET)~\cite{bowman2021elapsed} are proposed.
These alternatives can be generalized as Proof-of-X (PoX), where X represents the used resource to thwart the Sybil attack~\cite{sybil2002}. Most of the PoX-based blockchain protocols~\cite{bentov2016snow, Parity,bowman2021elapsed} utilize NC to make participants reach the consensus of ordered transactions such that they also suffer from block withholding.
Therefore, the idea of Crystal can also be applied to these protocols to prevent block withholding behaviors. 
Also note that some PoX-based protocols such as e-Pos~\cite{epos} or Algorand~\cite{Gilad2017}, which adopt Byzantine Fault Tolerant (BFT) consensus rather than NC, do not have block withholding issues. 
This is because the voting process in BFT consensus enforces a block to be published before it can be appended to the blockchain. 
}

\subsection{\jn{Beyond Synchronous Network}}
\jn{In this work, the security analysis (e.g., safety and liveness properties) of Crystal is done under the synchronous network assumption, i.e., there is a delay bound $\Delta$ for honest participants' messages. 
This is also known as the non-lock-step synchronous settings, which have been widely used in previous works~\cite{Pass2017, Kiffer2018, ren2019analysis, pass2017sleepy, Jing2021, niu2019analysis, Race2020}.
However, some measurement studies of the Bitcoin network show that some honest participants may not synchronize on a newly published block due to network attacks~\cite{Saad2019, Hijacking, eclipse2015}, unreachable participants as well as inefficient protocol design~\cite{SaadRoot, Saad2021}.
In other words, the synchronous network assumption of Bitcoin may not hold in the above cases. This also works for Crystal since it has the same network setting as Bitcoin. 
Therefore, for comprehensive security analysis, it is better to extend the synchrony model in Crystal to weakly synchronous or asynchronous network models.} 

\jn{The extension of the analysis to weakly synchronous or asynchronous network models is challenging. This is because a weaker network guarantee not only worsens the existing attacks (e.g., double-spending attacks~~\cite{SyncAttack}), but also leads to new attacks (e.g., HashSplit attack in ~\cite{Saad2021}). More attack strategies make a formal end-to-end security analysis significantly complicated. 
Fortunately, there are several state-of-the-art studies to analyze the security of Bitcoin in these networks~\cite{Saad2021, SyncAttack, sankagiri2021longest, ameen2022blockchain}. Saad et al. in~\cite{Saad2021} analyze the safety and chain quality of Bitcoin in the asynchronous network. They also propose a new attack, called HashSplit, in which an attacker can create multiple forking branches to affect the chain quality property. 
In~\cite{SyncAttack}, the double-spending attack is revisited in a weakly synchronous network. 
Ameen et al. in~\cite{ameen2022blockchain} consider the impact of message losses on the security of Bitcoin, while Sankagiri et al. in~\cite{sankagiri2021longest} take the random communication delay into account. Since Crystal inherits most components from Bitcoin, we believe that these studies can also be adapted to Crystal. What is more, Crystal may have better resilience to these network-level attacks since the usage of QCs is not affected by the network conditions and still can enforce an attacker to publish its blocks. We leave the extended analysis of Crystal as one of our future works.}

\section{Related Work}\label{sec:related work}
In this section, we provide related work that aims to decrease the efficiencies of block withholding in selfish mining or double-spending attacks.
Specifically, these protocols can be divided into three classes.
The first one is to modify the fork choice rule of NC; the second one is to change the reward mechanism; the third one is to redesign the consensus protocol.

\begin{table}[t]
    \centering
    \caption{\textbf{Comparisons with related designs.}}

    \begin{tabular}{@{}llll@{}}
        \toprule[1pt]
        Protocol                             & \tabincell{l}{Mitigate  selfish                   \\mining}  & \tabincell{l}{Thwart \\ double-spend}  & \tabincell{l}{Safety \\ Proof}  \\
        \midrule
        Bitcoin~\cite{nakamoto2012bitcoin}   & \xmark                          & \xmark & \cmark \\

        FruitChain~\cite{Pass2017fruit}      & \cmark$^{\lozenge}$             & \xmark & \cmark \\

        Perish or Publish~\cite{perishzhang} & \cmark                          & \xmark & \xmark \\

        Bobtail~\cite{Bissias2017BobtailAP}  & \cmark                          & \cmark & \xmark \\

        StrongChain~\cite{strongchain}       & \cmark                          & \cmark & \xmark \\

        Prism~\cite{pass2017hybrid}          & \xmark                          & \cmark & \cmark \\
        \hline
        \textbf{Crystal}                     & \cmark                          & \cmark & \cmark \\
        \bottomrule[0.9pt]
    \end{tabular}
    \begin{tablenotes}
        \item $\lozenge$ Zhang and Preneel showed that FruitChain is vulnerable to the selfish mining attack with reasonable parameters~\cite{Zhang2019CommonMetrics}.
    \end{tablenotes}
    \label{table:comparision}
\end{table}

Zhang and Preneel suggested a new fork-resolving policy, called Publish or Perish in~\cite{perishzhang}, which selectively neglects blocks that are not published in time and resolve forks by chain weight (rather than the length). 
The analysis shows the effectiveness of thwarting the selfish mining attack, but it is hard to detect an abnormal delay in a permissionless blockchain network that does not have reliable synchronized clocks.
Besides, honest blocks may also be delayed due to bad network connections or network attacks, and the rejection of these blocks may bring new economic losses to honest participants.
Szalachowski et al.~\cite{strongchain} proposed a new protocol, called StrongChain, which enables participants to publish weak solutions, i.e., solutions with lower mining targets.
Participants can include weak solutions in their blocks, and always select the chain with the largest weighted count of blocks and weak solutions to mine on.
Similar to StrongChain, Bissias and Levine~\cite{Bissias2017BobtailAP} proposed Bobtail, which enables participants to publish and collect all PoW solutions with lower difficulties until the mean of the $k$ smallest hashes satisfies a certain difficulty.
Although both StrongChain and Bobtail showed that they can thwart the selfish mining and double-spending attacks through simulations, it is unclear whether an adaptive attacker can invent protocol-specific attacks to undermine their security.
The lack of formal end-to-end proof for safety makes us worry about their security.
For example, we found that when the ratio of blocks with weak solutions is high in StrongChain, there may have a balance attack that will break the safety~\cite{yu2018ohie}.

The second type of protocol is to change the reward mechanism to defend against incentive-based attack~\cite{Pass2017fruit,wood2014ethereum}.
Ethereum~\cite{wood2014ethereum} introduces a new reward, called uncle rewards, for these blocks that are not included in the longest chain.
However, studies show that Ethereum’s uncle rewards raise the selfish mining profit and lower the computation power threshold to perform the attack~\cite{niu2019selfish}.
Pass and Shi~\cite{Pass2017fruit} proposed Fruitchains, which distributes rewards to all recent fruits that are parallel products of block mining. The authors proved that as long as some parameters are well-chosen, a fruit ignored by an attacker block will still be added to the blockchain by some honest blocks, therefore the protocol can guarantee fair reward distribution. However, the protocol can only defend the selfish mining attack but is not effective for other attacks related to block withholding.

The third type is to redesign the consensus protocol~\cite{bagaria2019deconstructing}.
Prism~\cite{bagaria2019deconstructing} deconstructs the functions of a block in NC into three types of blocks: proposer blocks, transaction blocks, and voter blocks.
Besides, by consisting of multiple parallel chains, Prism can greatly improve the throughput and reduce transaction latency by thwarting double-spending attacks.
However, Prism still adopts the longest chain rule, which implies that it does not defend against the selfish mining attack.
Besides, Prism does not provide any reward design, so whether it is secure against incentive-based attacks is unknown.

\section{Conclusion and Future Work} \label{sec:conclusion}
In this paper, we proposed a new transparent NC-style protocol, called Crystal, which leverages quorum certificates to thwart block withholding.
We have shown that Crystal can effectively mitigate selfish mining and double-spending attacks. We also provided end-to-end formal safety proofs for Crystal, which can guarantee that Crystal does not introduce new serious attacks.
Besides, we have evaluated the overhead and the additional delay caused by quorum certificates.
We have two directions for future work.
\jn{First, we can extend the synchronous network model by considering message losses and non-uniform propagation delay. The extension ensures a more accurate security analysis of Crystal.}
Second, we may apply the idea behind Crystal to other PoX-based blockchains.

\bibliographystyle{IEEEtran}
\bibliography{reference}

\appendices
\section{Concentration Bounds}
In this section, we provide the concentration bounds that we use in our analysis. We denote the probability of an event $E$ by $\Pr[E]$ and the expected value of a random variable $X$ by $\e{X}$. We will use the following bounds in our proofs.

\begin{lemma}[Chernoff bound] \label{lem:chernoff}
    Let $X_1$, . . . ,$X_n$ be independent random variables, and let $\mu := \e{\sum_{i = 1}^{n}X_i}$. For $0 < \delta < 1$, we have $ \Pr[\sum_{i = 1}^{n}X_i >(1 + \delta)\mu] < e^{-\Omega(\delta^2 \mu)}$ and $\Pr[\sum_{i = 1}^{n}X_i <(1- \delta)\mu] < e^{-\Omega(\delta^2 \mu)}$
\end{lemma}

\begin{lemma}[Poisson tail] \label{lem:poisson} Let $X$ be a Poisson random variable with rate $\mu$ (which is also its expectation). For $0 < \delta < 1$, we have $\Pr[X > (1+\delta)\mu] < e^{-\Omega(\delta^2 \mu)}$ and $\Pr[X > (1-\delta)\mu] < e^{-\Omega(\delta^2 \mu)}$.
\end{lemma}

\begin{lemma} [Chernoff bound for dependent random variables \cite{niu2019analysis}] \label{lem:key_step}
    Let $T$ be a positive integer. Let $X^{(j)} = \sum_{i = 0}^{n-1} X_{j + iT}$ be the sum of $n$ independent indicator random variables and $\mu_j = \e{ X^{(j)} }$ for $j \in \{1, \ldots, T\}$. Let $X = X^{(1)} + \cdots + X^{(T)}$. Let $\mu = \min_j \{ \mu_j \}$. Then, for $0 < \delta < 1$, $\Pr\left[ X \le (1 - \delta) \mu T \right] \le e^{-\Omega(-\delta^2 \mu / 2)}$.
\end{lemma}

\section{Selfish Mining Analysis} \label{appen: selfish}
\vspace{1mm} \noindent \textbf{Theorem 1.} \emph{ For $\alpha \in [0, 0.4)$, the attacker in Crystal cannot gain a higher fraction of blocks included in the longest chain than $\alpha$ in the long term.
}

\begin{proof}
    First, by solving the above Markov model, we can obtain the steady-state distribution of each state as follows:
    \begin{equation} \nonumber
        \begin{split}
            \pi_0 = \frac{\beta}{1 + \alpha \beta}, \quad \pi_{1} = \frac{\alpha}{1 + \alpha \beta} , \quad
            \pi_{0^{\prime}} = \frac{\alpha \beta}{1 + \alpha \beta}.
        \end{split}
    \end{equation}
    Next, we can analyze the blocks obtained by the attacker and by the honest participants in each state transition, respectively. In particular, we focus on the blocks that eventually end up in the main chain. We detail the blocks on each event as follows:
    \begin{itemize} [leftmargin=*]
        \item \emph{Case $a$: $S_0 \xrightarrow{\beta} S_0$.} In this case, honest participants generate one block and the attacker adopt it, which happens with probability $\beta$;

        \item \emph{Case $b$: $S_{0^{\prime}} \xrightarrow{\alpha} S_0$.} The attacker mines another block on its forking branch, which happens with probability $\alpha$. All participants will accept the attacker's branch due to the longest chain rule, and so the attacker has two blocks in the main chain.

        \item \emph{Case $c$: $S_{0^{\prime}} \xrightarrow{\gamma \beta}S_0$.} Some honest participants mine the next block on the attacker's forking branch, which happens with probability $\gamma \beta$. Both the attacker's block and the new honest block are included in the main chain.

        \item \emph{Case $d$: $S_{0^{\prime}} \xrightarrow{(1-\gamma)\beta}S_0$.} Some honest participants mine the next block on the honest participants' fork branch, which happens with probability $(1 - \gamma) \beta$. The attacker will accept these two blocks and mine on top of them.
    \end{itemize}
    Thus, we can obtain the blocks mined by the attacker and honest participants
    \begin{equation}
        r_a = \overset{\text{case $b$}}{\overbrace{\pi_{0^{\prime}} \cdot \alpha \cdot 2 }} + \overset{\text{case $c$}}{\overbrace{\pi_{0^{\prime}} \cdot \gamma \beta \cdot 1 }} = \frac{\alpha\beta(2\alpha+\gamma\beta)}{1+\alpha\beta},
    \end{equation}

    \begin{equation}
        \begin{split}
            r_{others} &= \overset{\text{case $a$}}{\overbrace{\pi_{0} \cdot \beta \cdot 1 }} +  \overset{\text{case $c$}}{\overbrace{\pi_{0^{\prime}} \cdot \gamma \beta \cdot 1 }} + \overset{\text{case $d$}}{\overbrace{\pi_{0^{\prime}} \cdot (1-\gamma) \beta \cdot 2 }} \\
            &= \frac{\beta^2 (1 + 2\alpha- \gamma\alpha)}{1+\alpha\beta}.
        \end{split}
    \end{equation}
    Adding these up, we can obtain that the fraction of blocks in the main chain created by the attacker is
    \begin{equation} \label{eq:selfish}
        R_a = \frac{r_a}{r_a + r_{others}} = \frac{\alpha (2\alpha + \gamma \beta)}{\beta + 2\alpha} = \alpha + \frac{\alpha(1-\alpha)(\gamma -1)}{1+\alpha}.
    \end{equation}
    For the last step, note that $\beta = 1 - \alpha$. As $0 \leq \gamma \leq 1$ and $0 \leq \alpha \leq 1$, the item $\frac{\alpha(1-\alpha)(\gamma -1)}{1+\alpha} \leq 0$. Hence, the fraction of blocks in the main chain owned by the attacker after launching a selfish mining attack is no larger than $\alpha$. The proof is done.
\end{proof}

\section{Safety Proofs} \label{appen: safety proof}

\vspace{1mm} \noindent \textbf{Lemma 2.} \emph{If an honest block $B_{\ell}$ is mined at the time $t_0$, then every honest participant can start to mine a block of at least height $\ell+1$ by the time $t_0 + 2\Delta$.}

\begin{proof}
    First, this honest block $B_{\ell}$ will reach all the honest participants by time $t_0 + \Delta$.
    Second, its parent block (no matter honest or malicious) will reach all the honest participants by time $t_0 + \Delta$. This argument applies to all of its ancestor blocks.
    Then, all honest committee members on the committee $\mathcal{M}_{B_{\ell}}$ will sign the digest of this block and broadcast their signatures.
    Finally, by the time $t_0 + 2\Delta$, every honest participant will observe a chain consisting of this block, its ancestor blocks, and/or other (honest or malicious) blocks mined on top of this block.
    If there are no such new blocks, the chain length is $\ell$. Otherwise, the chain length is greater than $\ell$.
    Every honest participant can start to mine for blocks with at least a height of $\ell+1$.

    \vspace{1mm} \noindent \textbf{Lemma 3.} Suppose block $B$ is a converged block of height $\ell$, then the block $B$ is the only honest block of height $\ell$.

    \begin{proof}
        Suppose for contradiction that two honest blocks $B$ and $B'$ of height $\ell$ are mined at time $t$ and $t'$ respectively.
        Since no other honest block is mined between time $t-2\Delta$ and $t+2\Delta$, we have have $t' \ge t+2\Delta$ or $r' \le t-2\Delta$. If $t' \ge t + 2\Delta$, by Lemma~\ref{lem:delta}, every honest participant observes a chain of length at least $\ell$ by time $t'$ and mines top of it. Therefore, no honest participant will mine a new block of height $\ell$ after time $t'$, leading to a contradiction.
        Similarly, if $t' \le t-2\Delta$, every honest participant observes a chain of length at least $\ell$ before the time $t$ (or even earlier), leading to a contradiction.
    \end{proof}

    Next, we prove the bounds on the number of converged blocks in a time interval of $t$.
    To achieve this, we first apply the Chernoff bound to obtain a bound on the number of blocks mined in the given interval, and then use Chernoff again to bound how many of these blocks are converged.
    The proof follows some ideas from prior works~\cite{niu2019analysis}. \\

    \vspace{1mm} \noindent \textbf{Lemma 4.} Let $\eta = e^{-2 \beta f \Delta }$. For any $0 < \delta <1$, the number of converged blocks mined in a time interval $t$ is at least $(1+\delta) \eta^2  \beta f t $, except with probability $e^{-\Omega\left(\delta^2 \eta^2 \beta f t \right)}$.

    \begin{proof}
        Let $N_H(t)$ denote the number of honest block mined in a time interval $t$, and note that $\e{N_H(t)} = \beta f t$.
        Then, for any $\delta_1 \in (0, 1)$, we have $\Pr[N_H(t) \leq (1 - \delta_1)\beta f t] \leq e^{-\Omega\left(\delta_1^2 \beta f t\right)}$ by Lemma~\ref{lem:chernoff}.
        In particular, let $n = (1 - \delta_1)\beta ft$ be an integer by choosing a suitable $t$.
        We enumerate the first $n$ honest blocks mined since the start of the time interval as blocks $1, 2, . . . , n$.
        Without loss of generality, we assume there is a block $0$ (resp. block $n+1$) that is the last honest block mined before (resp. after) the interval.

        Let $X_i$ denote the block interval between $(i-1)$-th and $i$-th block. Recall the mining process of honest participants is the Poisson process with rate $\beta f$. Hence, $X_i$ follows i.i.d. exponential distribution with the same rate.
        Let $Y_i$ denote an indicator random variable which equals one if the $i$-th block is converged block and zero otherwise. Define $Y = \sum_{i=1}^{n}Y_i$.
        It is easy to see that the $i$-th block is a converged block if $X_i \geq 2\Delta$ and $X_{i+1} \geq 2\Delta$. Since $X_i$ and $X_{i+1}$
        are independent, we have $\Pr[Y_i = 1] = \Pr[X_i \geq 2\Delta]\Pr[X_{i+1} \geq 2\Delta] = e^{-4 \beta f \Delta} = \eta^2$.
        Note that $Y_i$ and $Y_{i+1}$ are not independent since they both depend on the event that $X_{i+1} \geq 2\Delta$, but $Y_i$ and $Y_{i+2}$ are independent.
        Thus, $Y$ can be broken up into two summations of independent Boolean random variables $Y = \sum_{odd}{Y_i} + \sum_{even}{Y_i}$.
        By Lemma~\ref{lem:key_step}, we have $\Pr[Y \leq (1-\delta) \eta^2 \beta ft] \leq e^{-\Omega\left(\delta^2 \eta^2 \beta f t \right)}$.
    \end{proof}

    \begin{lemma}\label{lemma:converged2}
        Suppose $\eta^2 \beta > (1+\delta) \alpha$. In a time interval $t$, the number of converged blocks is greater than the number of malicious blocks except for $e^{-\Omega\left(\delta^2 \min\{ \eta, \beta f \Delta  \}  m \right)}$ probability.
    \end{lemma}

    \begin{proof}
        Let $N(t)$ (resp. $N_A(t)$) denote the converged (resp. malicious) block mined in the interval $t$.
        By Lemma~\ref{lemma:converged}, the number of converged blocks is $N(t) > (1-\delta_1) \eta^2 \beta f t$ except for $e^{-\Omega\left(\delta_1^2 \eta^2 \beta f t \right)}$ probability.
        Similarly, by Lemma~\ref{lem:poisson}, $N_A(t) < (1 + \delta_2) \alpha f t$ except for $e^{-\Omega\left(\delta_2^2 \alpha f t\right )}$ probability.
        By setting $\delta_1 = \delta_2 = \delta/4$ and noticing $\frac{1 + \delta/4}{1 - \delta/4} < 1 + \delta$, we have
        $(1 - \delta_1) \eta^2 \beta > (1 + \delta_2) \alpha$.
        Therefore, $N(t) > N_A(t)$ except for $e^{-\Omega\left(\delta^2 \min\{ \eta^2 \beta, \alpha \}  ft \right)}$ probability.
    \end{proof}

    \vspace{1mm} \noindent \textbf{Theorem 3} (Safety)\textbf{.} Suppose $\eta^2 \beta > (1+\delta) \alpha$.
    If $B$ and $B'$ are two distinct blocks of the same height, then they cannot be
    both confirmed, each by an honest participant. This property holds, regardless of malicious action, except for $e^{-\Omega\left(\delta^2 \min\{ \eta^2 \beta, \alpha \} k \right)}$ probability.

    \begin{proof}
        Consider the event $\mathcal{E}$ that ``$B$ and $B'$ of the same height are both confirmed, each by an honest participant." We will show that
        this event happens with probability at most $e^{-\Omega\left(\delta^2 \min\{ \frac{\eta}{f \Delta}, \beta   \} k \right)}$, regardless
        of malicious action.
        Let $t_1$ (resp., $t_1'$) be the smallest time when $B$ (resp., $B'$) is confirmed. Without loss of generality, we assume that $t_1 \ge t_1'$. Let $B_1$ be the most recent ancestor of $B$ and $B'$. That is, there are two disjoint subchains mined on top of $B_1$, one containing $B$ and the other containing $B'$. Let $B_0$ be the most recent \emph{honest} ancestor of
        $B$ and $B'$. Note that $B_0$ can be $B_1$ (if $B_1$ is honest) or the genesis block. Suppose that $B_0$ is mined (by some honest participant) at time $t_0$. For convenience, we assume that the genesis block is mined at time $0$. This makes $t_0$ well defined. We next define the following two events:
        \begin{itemize}[leftmargin=*]
            \item $\mathcal{E}_1(t_0, t_1)$: At time $t_1$, there are two disjoint subchains mined on top of $B_1$, each containing at least $k$ blocks mined from time $t_0$ to time $t_1$;
            \item $\mathcal{E}_2(t_0, t_1)$: $N_H(t) \le N_A(t)$, where $t = t_1 - t_0$. Here, $N_H(t)$ (resp. $N_A(t)$) is the number of converged (resp. malicious) blocks mined in time $[t_0, t_1]$.
        \end{itemize}
        We will show that $\mathcal{E} \subseteq \mathcal{E}_1(t_0, t_1)  \subseteq \mathcal{E}_2(t_0, t_1)$, regardless of
        malicious action.

        \begin{itemize}[leftmargin=*]
            \item $\mathcal{E} \subseteq \mathcal{E}_1(t_0, t_1)$: At time $t_1$, one subchain contains $B$ as well as $k-1$ blocks mined on top of $B$ (due to the $k$-deep confirmation rule). Similarly, the other subchain contains $B'$ as well as $k-1$ other blocks on top of $B'$. These blocks cannot be mined before $t_0$, because $B_0$ is an honest block.

            \item $\mathcal{E}_1(t_0, t_1) \subseteq \mathcal{E}_2(t_0, t_1)$: We will show that whenever there is a converged block of height $\ell$ mined in time $[t_0, t_1]$, there must be a ``matching"  malicious block of height $\ell$ mined between time $t_0$ and time $t_1$.
                  To see this, suppose that a converged block $B^*$ without a matching malicious block. By Lemma~\ref{lem:delta}, $B^*$ has a larger height than $B_0$.
                  On the one hand, if $B^*$ has a smaller height than $B$,
                  then $B^*$ must be an honest ancestor of $B$ and $B'$. This is because $B^*$ is the only block at its height. This contradicts the fact that $B_0$ is the most recent honest ancestor.
                  On the other hand, if $B^*$ has a larger height than $B$, then both subchains will contain $B^*$ at time $t_1$.
                  This is because $B^*$---the only block at its height---will reach all the honest participants by time $t_1$. As a result,
                  the subchain with $B$ will certainly contain $B^*$ and so the height of $B^*$ is at most the height of $B$ plus $k-1$. Similarly, the subchain with $B'$ will contain $B^*$, since there are at least $k-1$ blocks on top of $B'$. This leads to a contradiction.
        \end{itemize}
        By (a slightly modified version of) Lemma~\ref{lemma:converged2}, for any given $t = t_1 - t_0$, we have
        \begin{equation}
            \Pr(\mathcal{E}_2(t_0, t_1)) \le e^{-\Omega\left(\delta^2 \min\{ \eta^2 \beta, \alpha  \} ft  \right)}.
        \end{equation}
        Finally, we will bound $t$ and complete the proof. We claim that
        \begin{equation}
            t > \frac{2k}{(1+\delta) f}
        \end{equation}
        except for $e^{-\Omega\left(\delta^2 k \right)}$ probability, regardless of malicious action.
        To see this, recall that $\mathcal{E}_1(t_0, t_1)$ states that two subchains contain at least $2k$ blocks. Hence, $t$ is smallest if all the mined blocks $N(t)$ from time $t_0$ to time $t_1$ (the number of which is honest blocks and malicious blocks) belong to
        these two subchains. By Lemma~\ref{lem:poisson},
        \begin{equation}
            \Pr\left( N(t) \ge (1+\delta) f t  \right) \le e^{-\delta^2 f t/3}.
        \end{equation}
        So, if we set $t = \frac{2k}{(1+\delta) f }$, then we have
        \begin{equation}
            \Pr\left( N(t) \ge 2k \right)  \le e^{-\delta^2 \frac{(2k)}{1 + \delta}/3}.
        \end{equation}
        This proves our claim.

        Define the event $\mathcal{D}$ as $t > \frac{2k}{(1+\delta) f }$.
        Then, $\Pr(\mathcal{D}^c) \le e^{-\Omega\left(\delta^2 k \right)}$, where $\mathcal{D}^c$ is the complement of $\mathcal{D}$. Therefore, for any malicious action, we have
        \begin{align}
            \Pr(\mathcal{E}) & = \Pr(\mathcal{D}^c) \Pr(\mathcal{E} | \mathcal{D}^c) + \Pr(\mathcal{D}) \Pr(\mathcal{E} | \mathcal{D}) \\
                             & \le \Pr(\mathcal{D}^c) + \Pr(\mathcal{D}) \Pr(\mathcal{E}_2(r_0, r) | \mathcal{D})                      \\
                             & \le \Pr(\mathcal{D}^c) + \Pr(\mathcal{E}_2(r_0, r) | \mathcal{D})                                       \\
                             & \le e^{-\Omega\left(\delta^2 \min\{ \frac{\eta}{f \Delta}, \beta   \} k \right)}
        \end{align}
        where the last inequality follows from $k \ge \min\{ \frac{\eta}{f \Delta}, \beta   \} k$.
    \end{proof}
    Note that in the safety proof of NC~\cite{garay2015bitcoin,Pass2017}, a block is confirmed if it is in the longest chain, and there are $k$ blocks mined on it rather than $k-1$ blocks adopted here.
    We make the subtle change to make the $k$-confirmation rule consistent with the one usually adopted in double-spending attack analysis.
\end{proof}

\end{document}